\documentclass[11pt]{article}

\usepackage{dsfont}

\usepackage{amsthm,amsfonts,amsmath,amssymb,epsfig,color,float,graphicx,verbatim, enumitem}
\usepackage{multirow}
\usepackage{algorithm}
\usepackage[noend]{algpseudocode}
\usepackage{enumitem}

\newif\ifhyper\IfFileExists{hyperref.sty}{\hypertrue}{\hyperfalse}
\hypertrue
\ifhyper\usepackage{hyperref}\fi

\usepackage{enumitem}

\usepackage{framed}
\usepackage{nicefrac}

\newtheorem{theorem}{Theorem}[section]

\newtheorem{cond}[theorem]{Condition}
\newtheorem{lemma}[theorem]{Lemma}
\newtheorem{informal theorem}[theorem]{Theorem (informal statement)}

\newtheorem{proposition}[theorem]{Proposition}

\newtheorem{claim}[theorem]{Claim}
\newtheorem{fact}[theorem]{Fact}

\theoremstyle{definition}
\newtheorem{definition}[theorem]{Definition}
\newcommand{\eqdef}{\stackrel{{\mathrm {\footnotesize def}}}{=}}

\newcommand{\bX}{\mathbf{X}}
\newcommand{\bY}{\mathbf{Y}}
\newcommand{\bx}{\mathbf{x}}
\newcommand{\by}{\mathbf{y}}

\newcommand{\p}{\mathbf{P}}
\newcommand{\q}{\mathbf{Q}}

\newcommand{\R}{\mathbb{R}}

\newcommand{\Z}{\mathbb{Z}}

\newcommand{\E}{\mathbf{E}}
\newcommand{\eps}{\epsilon}
\newcommand{\dtv}{d_{\mathrm{TV}}}

\renewcommand{\Pr}{\mathbf{Pr}}

\newcommand{\var}{\mathbf{Var}}

\newcommand{\bin}{\mathrm{Bin}}
\newcommand{\is}{\mathrm{IS}}

\newcommand{\D}{\mathcal{D}}

\newcommand{\wt}{\widetilde}
\newcommand{\wh}{\widehat}

\author{
Ilias Diakonikolas\thanks{Supported by NSF Medium Award CCF-2107079,
NSF Award CCF-1652862 (CAREER), a Sloan Research Fellowship, and
a DARPA Learning with Less Labels (LwLL) grant.}\\
University of Wisconsin Madison\\
{\tt ilias@cs.wisc.edu}\\
\and
Daniel M. Kane\thanks{Supported by NSF Medium Award CCF-2107547,
NSF Award CCF-1553288 (CAREER), a Sloan Research Fellowship, and a grant from CasperLabs.}\\
University of California, San Diego\\
{\tt dakane@cs.ucsd.edu}\\
\and
Yuxin Sun \thanks{Supported by NSF Award CCF-1652862 (CAREER).}\\
University of Wisconsin Madison\\
{\tt yxsun@cs.wisc.edu}\\
}

\title{Optimal SQ Lower Bounds for Robustly Learning\\
Discrete Product Distributions and Ising Models}

\begin{document}

\maketitle

\begin{abstract}
We establish optimal Statistical Query (SQ) lower bounds
for robustly learning certain families of discrete high-dimensional distributions. 
In particular, we show that no efficient SQ algorithm with access to an $\eps$-corrupted 
binary product distribution can learn its mean within $\ell_2$-error $o(\eps \sqrt{\log(1/\eps)})$. 
Similarly, we show that no efficient SQ algorithm with access to an $\eps$-corrupted 
ferromagnetic high-temperature Ising model can learn the model 
to total variation distance $o(\eps \log(1/\eps))$. 
Our SQ lower bounds match the error guarantees of known algorithms for these problems, 
providing evidence that current upper bounds for these tasks are best possible.
At the technical level, we develop a generic SQ lower bound for discrete
high-dimensional distributions starting from low-dimensional moment matching constructions
that we believe will find other applications. Additionally, we introduce new ideas
to analyze these moment-matching constructions for discrete univariate distributions.
\end{abstract}

\setcounter{page}{0}

\thispagestyle{empty}

\newpage

\section{Introduction}\label{ssec:intro}

\subsection{Background and Motivation} \label{ssec:background-results}

\paragraph{Robust Statistics and Information-Computation Tradeoffs}
We study high-dimensional learning in the presence of a constant
fraction of arbitrary outliers. Robust learning in high dimensions
has its roots in robust statistics, a branch of
statistics initiated in the 60s 
with the pioneering works of Tukey and Huber~\cite{Tukey60, Huber64}.
Early work 
developed minimax optimal estimators for various robust estimation tasks, albeit
with runtimes exponential in the dimension.
A recent line of work in computer science,
starting with~\cite{DKKLMS16, LaiRV16}, developed polynomial time
robust estimators for a range of high-dimensional statistical tasks.
Algorithmic high-dimensional robust statistics is by now a relatively mature field,
see, e.g.,~\cite{DK19-survey, DiakonikolasKKL21} for recent overviews.

The line of work in algorithmic robust statistics
established the existence of computationally efficient algorithms with
dimension-independent error guarantees for a range of high-dimensional
robust estimation tasks. In some instances,
these algorithms achieve the information-theoretically
optimal error (within constant factors). Alas, in several interesting settings,
there is a super-constant gap between the information-theoretic optimum
and what known efficient algorithms achieve.
This raises the following natural question:
{\em For a given high-dimensional robust estimation task,
is the information-theoretically optimal error achievable in polynomial time?}

In several high-dimensional statistical settings, there is strong evidence that
inherent resource tradeoffs exist. In robust statistics,
the study of such {\em information-computation tradeoffs}
was initiated in~\cite{DKS17-sq}, which established the first such lower bounds
in the  Statistical Query (SQ) model~\cite{Kearns:98}.
The methodology for proving such lower bounds introduced in~\cite{DKS17-sq}
applies to ``Gaussian-like'' distributions.
In particular, the general problem underlying that work --- known as
Non-Gaussian Component Analysis (NGCA)~\cite{JMLR:blanchard06a, TanV18, GoyalS19} ---
considers distributions that are distributed as a standard Gaussian
in all but one hidden direction. 
The methodology introduced in~\cite{DKS17-sq} has led to
SQ lower bounds for a range of statistical problems, 
including robust mean and covariance estimation~\cite{DKS17-sq},
robust sparse mean estimation~\cite{DKS17-sq},
adversarially robust learning~\cite{BPR18},
robust linear regression~\cite{DKS19}, list-decodable estimation~\cite{DKS18-list, DKPPS21},
learning simple neural networks~\cite{DiakonikolasKKZ20}, and robust supervised learning
in a variety of noise models~\cite{DKZ20, DiakonikolasKPZ21, DK20-Massart-hard, DKKTZ21-benign}.

Here we are interested in exploring information-computation tradeoffs for robustly learning
{\em discrete} high-dimensional distributions. The two concrete examples --- that were the main motivation
for this work --- are (1) the class of binary product distributions, and (2) the (more general) class of Ising models.
For both of these distribution classes, there are gaps between the information-theoretically optimal
error and the error that known polynomial-time algorithms can achieve.
Given the aforementioned prior work for Gaussian-like distributions~\cite{DKS17-sq},
it would be tempting to conjecture that these gaps are in fact inherent. In this work, we develop
the necessary methodology that allows us to prove such statements for discrete distributions,
and in particular for the aforementioned families.

Before we proceed, we give the necessary background on the SQ model and robust statistics.

\paragraph{Statistical Query (SQ) Model}
SQ algorithms are the class of algorithms
that are only allowed to query expectations of bounded functions
of the underlying distribution rather than directly access samples.
The SQ model was introduced by Kearns~\cite{Kearns:98} in the context of supervised learning
as a natural restriction of the PAC model~\cite{val84} and has been extensively studied
in learning theory. A recent line of work~\cite{FGR+13, FeldmanPV15, FeldmanGV17, Feldman17}
generalized the SQ framework for search problems over distributions.

The class of SQ algorithms is fairly broad:  a wide range of known algorithmic techniques
in machine learning are known to be implementable in the SQ model.
These include spectral techniques, moment and tensor methods,
local search, and many others
(see, e.g.,~\cite{Chu:2006, FGR+13, FeldmanGV17}).
A notable exception are learning algorithms using Gaussian elimination
(in particular for learning parities, see, e.g.,~\cite{BKW:03}).
Notably, \cite{BBHLS20} recently established a connection
between the SQ model and low-degree polynomial tests under certain assumptions.

\paragraph{Contamination Model}
We focus on the following contamination model,
where the adversary can corrupt the true distribution in total variation distance.

\begin{definition}[TV-contamination] \label{def:oblivious}
Given a parameter $0<\eps < 1/2$ and a distribution class $\mathcal{D}$, 
we say that a distribution $D'$ is an \emph{$\eps$-corrupted} version of a distribution 
$D\in\mathcal{D}$ if $\dtv(D,D')\leq \eps$.
\end{definition}

We will be interested in algorithms robust against this kind of contamination. 
In particular, we want algorithms that given sample access to a distribution $D'$ 
which is an $\eps$-corrupted version of some unknown distribution $D\in\mathcal{D}$, 
can approximate relevant parameters of the ``true'' distribution $D$. 
For such algorithms, one may want to consider the distribution $D'$ to be adversarially selected, 
perhaps in a way designed to fool the particular algorithm in question. 
We also note that several algorithms in robust statistics can be shown to succeed
in the presence of even stronger contamination models, 
such as the strong contamination model, 
where the adversary can inspect the clean samples drawn and adaptively choose 
which samples to corrupt and how. However, these stronger models are harder 
to formalize for SQ algorithms where our lower bounds will apply.

\subsection{Problems of Interest and Our Results} \label{ssec:results}

With this background, we are ready to summarize prior algorithmic work on the two problems
of interest and informally state our contributions.

\paragraph{Robust Mean Estimation for a Binary Product Distribution}
A binary product distribution is a distribution over $\{0, 1\}^M$ whose coordinates are independent.
We consider the algorithmic problem of computing an approximation 
to the mean vector $\mu_P$ of a binary product distribution $P$, in $\ell_2$-norm, 
given access to a set of samples from an $\eps$-corrupted version of $P$.
\cite{DKKLMS16} gave the first efficient algorithm for this problem that 
outputs an estimate $\wh{\mu}$ such that
with high probability $\|\wh{\mu} - \mu_P\| = O(\eps \sqrt{\log(1/\eps)})$. 
Information-theoretically, it is possible to approximate $\mu_P$ within $\ell_2$ error $\Theta(\eps)$.
Our first main result shows that this gap is inherent for SQ algorithms 
(see Theorem~\ref{thm:sq-bin-testing} for a detailed statement).

\begin{theorem}[SQ Lower Bound for Binary Products, Informal] \label{thm:bp-informal}
Any SQ algorithm that robustly learns the mean of a binary product distribution over $\{0, 1\}^M$, given access to 
an $\eps$-corruption, within $\ell_2$-error $o(\eps \sqrt{\log(1/\eps)})$ either requires at least $2^{M^{\Omega(1)}}$
many statistical queries or must make a query of accuracy inverse super-polynomial in $M$.
\end{theorem}

\noindent Theorem~\ref{thm:bp-informal} shows that no SQ algorithm can robustly approximate the mean 
of a binary product distribution to error $o(\eps \sqrt{\log(1/\eps)})$ with a sub-exponential in $M^{\Omega(1)}$ queries, 
unless using queries of very small tolerance -- that would require super-polynomially many samples in $M$
to simulate. In that sense, Theorem~\ref{thm:bp-informal} is an information-computation tradeoff 
for robust mean estimation of a binary product distribution within the class of SQ algorithms.

\paragraph{Robustly Learning a Ferromagnetic High-Temperature Ising Model}
Given a symmetric matrix $(\theta_{ij})_{i,j\in[M]} \in \R_+^{M \times M}$ 
with zero diagonal, a ferromagnetic Ising model is a distribution over $\{\pm1\}^M$ with mass function
$P_\theta(x)=\frac{1}{Z(\theta)}\exp \big((1/2)\sum_{i,j\in[M]}{\theta_{ij}x_ix_j}\big)$,
where $Z(\theta)$ is a normalizing constant.
We say that an Ising model lies in the high-temperature regime
if there is a universal constant $0<\eta<1$ such that $\max_{i\in[M]}{\sum_{j\ne i}|\theta_{ij}|}\le1-\eta$.
Here we would like an algorithm that given samples from an $\eps$-corrupted version 
of an unknown ferromagnetic high-temperature Ising model $P$, 
approximates $P$ in total variational distance.
Prior work \cite{DiakonikolasKSS21} gave the first efficient algorithm for this problem that outputs an estimate
$\wh{P}$ such that with high probability $\dtv(\wh{P}, P) = O(\eps \log(1/\eps))$.
On the other hand, the information-theoretically optimal error in total variation distance is $\Theta(\eps)$.
Our second main result shows that this gap is inherent for SQ algorithms 
(see Theorem~\ref{thm:sq-ising-testing} for a detailed statement).

\begin{theorem}[SQ Lower Bound for Ising Models, Informal] \label{thm:ising-informal}
Any SQ algorithm that robustly learns a ferromagnetic high-temperature Ising Model over $\{\pm1\}^M$, 
given access to an $\eps$-corruption, within total variation distance $o(\eps \log(1/\eps))$ 
either requires at least $2^{M^{\Omega(1)}}$ many statistical queries or must make 
a query of accuracy inverse super-polynomial in $M$.
\end{theorem}

\noindent Similarly, Theorem~\ref{thm:ising-informal} is an information-computation tradeoff 
for robust learning of an Ising model within the class of SQ algorithms.
In summary, for both of these problems, we show that known algorithms are essentially optimal
within the class of Statistical Query (SQ) algorithms.

\medskip

In addition to the aforementioned concrete applications 
(Theorems~\ref{thm:bp-informal} and~\ref{thm:ising-informal}), 
we develop a novel generic SQ lower bound construction
for discrete structured high-dimensional distributions (Proposition~\ref{prop:gen-sq-prop}) 
that we believe will find other applications.

\subsection{Technical Overview}\label{ssec:techniques}
Here we provide an outline of our approach and techniques. 
To prove our SQ lower bound in the discrete setting,
we need to develop a novel generic discrete SQ lower bound machinery
for distributions over $\{0,1\}^M$. At a high level, our construction 
resembles the lower bound construction of \cite{DFTWW15} (which applies in the context of supervised learning), 
adapting several ideas of \cite{DKS17-sq} from a Gaussian version of this problem.

In general, establishing an SQ lower bound for learning distributions in some class 
essentially boils down to proving lower bounds for the corresponding SQ dimension~\cite{FGR+13}. 
In our case, this amounts to finding large families of $\eps$-corrupted 
binary product distributions or $\eps$-corrupted Ising models 
that have pairwise small chi-squared inner product 
with respect to some given base distribution. 
We will select as a base distribution 
the uniform distribution over the hypercube.

To construct these distributions, we will adapt and generalize the techniques of \cite{DFTWW15}. 
In particular, we aim to find a single distribution $D_0$ of the appropriate type over $\{0,1\}^m$, 
for some $m$ substantially smaller than $M$, so that $D_0$'s low-degree Fourier coefficients vanish. 
One can then use $D_0$ to obtain many different distributions over $\{0,1\}^M$ 
by embedding it over some subset (chosen in one of many different ways) 
of the coordinates, and using the uniform distribution over the remaining coordinates. 
One can show (see Lemma \ref{lem:cor-disc}) that this allows one to produce 
many nearly orthogonal distributions.

This leaves us with the task of producing an appropriate distribution $D_0$. 
To achieve this, we take inspiration from the Gaussian regime \cite{DKS17-sq}. 
In particular, we simplify matters by considering only \emph{symmetric} distributions $D_0$. 
This means that $D_0$ is determined by a {\em one-dimensional} distribution $A$ --- 
specifically, the distribution over the sum of the coordinates of $D_0$. 
This distribution $A$ must be close in total variation distance to an appropriate one-dimensional 
version of either a binary product distribution or an Ising model, 
and must match several of its low-degree moments with the Binomial distribution.

In order to construct these one-dimensional distributions, 
we again borrow ideas from \cite{DKS17-sq}. 
We want to obtain a distribution $A$ close to some other distribution $B$ that matches its low-degree 
moments with the binomial. We will achieve this by starting with the distribution $B$ 
and modifying its probability mass function (pmf) over some appropriately chosen interval $I$. 
In particular, if we modify it by a degree-$k$ polynomial $p$ over $I$, 
there will be a unique choice of this polynomial that gives us some specified first $k$ moments. 
To establish correctness, we need to verify that the resulting polynomial $p$ is not too large 
(both to ensure that the resulting pmf is non-negative 
and to ensure that $A$ and $B$ are close in total variation distance). 
This can be shown via an explicit analysis 
involving Legendre polynomials (as is done by \cite{DKS17-sq}, in the continuous case) 
along with additional technical work required to show that the change 
to the discrete setting does not significantly affect things.

\section{Preliminaries} \label{sec:prelims}

\paragraph{Notation} For $n \in \Z_+$, we denote $[n] \eqdef \{1,\ldots,n\}$.
For two distributions $p,q$ over a probability space $\Omega$,
let $\dtv(p,q)=\sup_{S\subseteq\Omega}|p(S)-q(S)|$
denote the total variation distance between $p$ and $q$.
We use $\Pr[\mathcal{E}]$ and $\mathbb{I}[\mathcal{E}]$ for
the probability and the indicator of event $\mathcal{E}$.
For a real random variable $X$, we use $\E[X],\var[X]$ to denote the expectation
and variance of $X$, respectively.
For $n\in\Z_+$ and $0\le p\le1$, we use $\bin(n,p)$
to denote the Binomial distribution with parameters $n$ and $p$.

\paragraph{Properties of Legendre Polynomials} We record some properties of Legendre polynomials that we will need throughout this paper.
\begin{fact}[\cite{Szego:39}]\label{fact:legendre-poly}
The Legendre polynomials, $P_i(x)$, for $i\in\mathbb{Z}_+$, satisfy the following properties:
(i) $P_i(x)$ is a degree $i$-polynomial with $P_0(x)=1$ and $P_1(x)=x$.
(ii) $\int_{-1}^{1}P_i(x)P_j(x)dx=\frac{2\delta_{ij}}{2i+1}$ for all $i,j\ge0$.
(iii) $|P_i(x)|\le1$ for all $|x|\le1$.
(iv) $P_i(x)=(-1)^iP_i(-x)$.
(v) $P_i(x)=2^{-i}\sum_{j=0}^{\lfloor i/2\rfloor}(-1)^j\binom{i}{j}\binom{2i-2j}{i}x^{i-2j}$.
(vi) $|P_i(x)|\le(4|x|)^i$ for all $|x|\ge1$.
(vii) $\int_{-1}^{1}|P_i(x)|dx\le O(1/\sqrt{i})$.
(viii) $|P'_i(x)|\le O(i^2)$ for all $|x|\le1$.
\end{fact}

\paragraph{Ising Models}
We recall basic facts about Ising models, which will be used throughout this paper.
\begin{definition}[Ising Model] \label{def:ising}
Given a symmetric matrix $(\theta_{ij})_{i,j\in[M]} \in \R^{M \times M}$ with zero diagonal,
the Ising model distribution $P_\theta$ is defined as:
$P_\theta(\bx)=\frac{1}{Z(\theta)}\exp \big((1/2)\sum_{i,j\in[M]}{\theta_{ij}x_ix_j}\big),\forall \bx\in\{\pm1\}^M$,
where the normalizing factor $Z(\theta)$ is called the partition function.
We call the matrix $(\theta_{ij})_{i,j\in[d]} \in \R^{M \times M}$ the interaction matrix. In addition, we say that $P_\theta$ is ferromagnetic if $\theta_{ij}\ge0,\forall i,j\in[M]$.
\end{definition}
The following Dobrushin's condition for Ising models is a classical assumption
needed to rule out certain pathological behaviors.
This condition is standard in various areas, including statistical physics, computational biology, machine learning,
and theoretical CS~\cite{kulske2003concentration, gotze2019higher, dagan2020estimating, adamczak2019note, gheissari2018concentration, marton2015logarithmic}.
\begin{definition}[Dobrushin's Condition]\label{def:high-temp}
Given an Ising model $P_{\theta}$ with interaction matrix $(\theta_{ij})_{i,j\in[M]}$,
we say that it satisfies Dobrushin's condition, or lies in the high temperature regime,
if there is a constant $0<\eta<1$ such that $\max_{i\in[M]}{\sum_{j\ne i}|\theta_{ij}|}\le1-\eta$.
\end{definition}

\paragraph{Statistical Query Algorithms}
We will use the framework of Statistical Query (SQ) algorithms for problems over distributions
introduced in~\cite{FGR+13}.
Before we get into the formal statement of our generic discrete SQ lower bound, 
we formulate it as a decision problem as follows:

\begin{definition}[Decision/Testing Problem over Distributions]\label{def:decision}
Let $D$ be a distribution and $\D$ be a family of distributions over $\R^M$. We denote by $\mathcal{B}(\mathcal{D},D)$ the decision (or hypothesis testing) problem in which the input distribution $D'$ is promised to satisfy either (a) $D'=D$ or (b) $D'\in\mathcal{D}$, and the goal of the algorithm is to distinguish between these two cases.
\end{definition}

We define SQ algorithms as algorithms that do not have direct access to samples from the distribution,
but instead have access to an SQ oracle. We consider the following standard oracle.
\begin{definition}[$\mathrm{STAT}$ Oracle]\label{def:stat}
Let $D$ be a distribution on $\R^M$. A \emph{Statistical Query (SQ)} is a bounded function $f:\R^M\to[-1,1]$.
For $\tau>0$, the $\mathrm{STAT}(\tau)$ oracle responds to the query $f$ with a value $v$ such that $|v-\E_{X\sim D}[f(X)]|\le\tau$.
We call $\tau$ the \emph{tolerance} of the statistical query.
A \emph{Statistical Query (SQ) algorithm} is an algorithm whose objective is to learn some information about an unknown distribution $D$ by making adaptive calls to the corresponding $\mathrm{STAT}(\tau)$ oracle.
\end{definition}

To define the SQ dimension, we need
the following definition.
\begin{definition}[Pairwise Correlation] \label{def:pc}
The pairwise correlation of two distributions with probability mass functions
$D_1, D_2 : \{0,1\}^M \to \R_+$ with respect to a distribution with mass $D:\{0,1\}^M \to \R_+$,
where the support of $D$ contains the supports of $D_1$ and $D_2$,
is defined as $$\chi_{D}(D_1, D_2) + 1 \eqdef \sum_{x\in\{0,1\}^M} D_1(x) D_2(x)/D(x) \;.$$
We say that a set of $s$ distributions $\mathcal{D} = \{D_1, \ldots , D_s \}$ over $\{0,1\}^M$
is $(\gamma, \beta)$-correlated relative to a distribution $D$ if
$|\chi_D(D_i, D_j)| \leq \gamma$ for all $i \neq j$, and $|\chi_D(D_i, D_j)| \leq \beta$ for $i=j$.
\end{definition}

We are now ready to define the notion of SQ dimension.

\begin{definition}[SQ Dimension] \label{def:sq-dim}
For $\gamma ,\beta> 0$, a decision problem $\mathcal{B}(\mathcal{D},D)$, 
where $D$ is fixed and $\mathcal{D}$ is a family of distributions over $\{0,1\}^M$,
let $s$ be the maximum integer such that there exists
$\mathcal{D}_D \subseteq \D$ such that $\D_D$ is $(\gamma,\beta)$-correlated
relative to $D$ and $|\D_D|\ge s$.
We define the {\em Statistical Query dimension} with pairwise correlations $(\gamma, \beta)$
of $\mathcal{B}$ to be $s$ and denote it by $\mathrm{SD}(\mathcal{B},\gamma,\beta)$.
\end{definition}

\noindent The connection between SQ dimension and lower bounds is captured
by the following lemma.

\begin{lemma}[\cite{FGR+13}] \label{lem:sq-from-pairwise}
Let $\mathcal{B}(\D,D)$ be a decision problem, where $D$ is the reference distribution and $\D$ is a class of distributions over $\R^M$. For $\gamma, \beta >0$, let $s= \mathrm{SD}(\mathcal{B}, \gamma, \beta)$. Any SQ algorithm that solves $\mathcal{B}$ with probability at least $2/3$ requires at least $s \cdot \gamma /\beta$ queries to the
$\mathrm{STAT}(\sqrt{2\gamma})$ oracles.
\end{lemma}

We note that the hypothesis testing problem of Definition~\ref{def:decision}
may in general be information theoretically hard. In particular, if some distribution
$D'\in\mathcal{D}$ is very close to the reference distribution $D$,
it will be hard to distinguish between $D'$ and $D$.
On the other hand, if $D'$ is far from the reference distribution $D$ in total variation distance
for any $D'\in\D$, then one can straightforwardly reduce the hypothesis testing problem
to the problem of learning an unknown $D'\in\D$ to small accuracy.
For completeness, we defer the formal statement and proof to Appendix~\ref{ssec:learning-to-testing}.

\section{Generic Discrete SQ Lower Bound Construction} \label{ssec:generic-discrete}
We start with some basic definitions.
\begin{definition}[Characters] \label{def:char}
For a subset $T\subseteq[M]$ and $\bx \in \{0, 1\}^M$, we denote $\chi_T(\bx) = (-1)^{\sum_{i \in T}x_i}$.
For a distribution $\p$ over $\{0,1\}^M$, let $\wh{\p}(T) = \E_{\bX\sim\p}[\chi_T(\bX)]$.
\end{definition}
We will denote by $U_M$ the uniform distribution over $\{0,1\}^M$.
By Plancherel's identity, we have the following fact about the chi-squared inner product
in the discrete setting.
\begin{fact} \label{fact:chi}
For distributions $\p, \q$ over $\{0,1\}^M$, we have that
$1+\chi_{U_M}(\p, \q) = \sum_{T\subseteq [M]} \wh{\p}(T) \wh{\q}(T).$
\end{fact}
We will require the orthogonal polynomials under the Binomial distribution.
\begin{definition}[Kravchuk Polynomial~\cite{Szego:39}] \label{def:Krav}
For $k, m, x \in \Z_+$ with $0\le k,x\le m$, the Kravchuk polynomial $\mathcal{K}_k(x;m)$ is the univariate degree-$k$ polynomial
in $x$ defined by $\mathcal{K}_k(x;m) := \sum_{T\subseteq[m], |T|=k} \chi_T(\by)=\sum_{j=0}^k(-1)^j\binom{x}{j}\binom{m-x}{k-j}$,
where $\by$ has $x$ 1's and $m-x$ 0's.
\end{definition}
\begin{fact}[Orthogonality~\cite{Szego:39}]
Let $j,k,m\in\Z_+$. Then, $\E_{X\sim\bin(m,1/2)}[\mathcal{K}_j(X;m)\mathcal{K}_k(X;m)]=\mathbb{I}[j=k] \binom{m}{k}.$
In particular, if $k\ge1$, then $\E_{X\sim\bin(m,1/2)}[\mathcal{K}_k(X;m)]=0$.
\end{fact}

Our basic technique for producing near-orthogonal distributions
over the hypercube takes inspiration from~\cite{DFTWW15}.
They show that if one can construct a distribution $D$ over a small number of coordinates
whose degree up-to-$k$ Fourier coefficients agree with the uniform distribution,
then by taking embeddings of $D$ into the hypercube as a junta can provide many orthogonal distributions.
This leaves us with finding our moment-matching distribution $D$. Our basic idea will be to make $D$
a symmetric distribution, as this will simplify things substantially due to the added symmetry.
Essentially, $D$ will be defined by some distribution $A$ on $\sum x_i$.
This distribution $A$ will need to nearly match the first $k$ moments with the Binomial distribution $\bin(m,1/2)$.

We now formally define the high-dimensional distribution family that is the basis
of our discrete SQ lower bound construction.
\begin{definition} [High-Dimensional Hidden Junta Distribution] \label{def:p-hidden-junta}
Let $m, M \in \Z_+$ with $m<M$.
For a distribution $A$ on $[m]\cup\{0\}$ with probability mass function (pmf) $A(x)$
and a subset $S\subseteq [M]$ with $|S|=m$, consider the probability distribution over $\{0,1\}^M$,
denoted by $\p^A_S$, such that for $\bX \sim \p^A_S$ the distribution $(X_i)_{i \not\in {S}}$ is the uniform distribution on its support and the distribution
$(X_i)_{i \in S}$
is symmetric with $\sum_{i\in S} X_i$ distributed according to $A$.
Specifically, $\p^A_S$ is given by the pmf
$$\p^A_S(\bx) = 2^{-M+m} A\left(\sum_{i\in S}x_i \right)\binom{m}{\sum_{i\in S}x_i}^{-1}.$$
\end{definition}

We now define the hypothesis testing problem
which will be used throughout this paper:
\begin{definition}[Hidden Junta Testing Problem]\label{def:testing}
Let $m,M\in\Z_+$ with $M>m$ and $A$ be a one-dimensional distribution over $[m]\cup\{0\}$. In the $(A,M)$-Hidden Junta Testing Problem,
one is given access to a distribution $D$ so that either $H_0$: $D=U_M$,
$H_1$: $D$ is given by $\p_S^A$ for some subset $S\subseteq[M]$ with $|S|=m$,
where $\p_S^A$ denotes the hidden junta distribution corresponding to $A$.
One is then asked to distinguish between $H_0$ and $H_1$.
\end{definition}
Note that this is just the hypothesis testing problem $\mathcal{B}(\D,D)$ with $D=U_M$ and $\D=\{\p_S^A\}$.
The following condition describes the approximate moment-matching property
of the desired distribution $A$ with the Binomial distribution.

\begin{cond}[Approximate Moment-Matching]\label{cond:moments-disc}
Let $\nu>0$ and $k,m \in \Z_+$ with $k \leq m$.
The distribution $A$ on $[m]\cup\{0\}$ satisfies
$|\E_{X\sim A}[\mathcal{K}_t(X;m)]|\leq \nu$, for all $1\leq t \leq k$.
\end{cond}
In particular, if $A$ exactly matches the first $k$ moments with $\bin(m,1/2)$,
then we will have that $\E_{X\sim A}[\mathcal{K}_t(X;m)]=\E_{X\sim\bin(m,1/2)}[\mathcal{K}_t(X;m)]=0$,
for all $1\le t\le k$.

In order to prove SQ lower bounds for the above testing problem, 
one needs to find many sets $S$ for which the corresponding $\p_S^A$ are nearly orthogonal. 
For this, we show that it suffices to find many subsets $S$ 
whose intersections are pairwise much smaller than $m$.
In particular, we prove that if $|S\cap S'|=o(m)$,
then the corresponding inner product will be sufficiently small.
This makes our technique somewhat reminiscent of~\cite{DKS17-sq},
which proves lower bounds in the Gaussian setting, where their hard distributions
are equal to some moment-matching distribution $A$ in a hidden direction $v$
and are standard Gaussian in the orthogonal directions. \cite{DKS17-sq} shows that if two such distributions
have hidden-directions $u$ and $v$, then the chi-squared inner product of these distributions
is on the order of $|u^Tv|^d$, where $d$ is the number of matching moments.
A significant difference with the Gaussian case here is in the way
we embed the one-dimensional distribution $A$ as a higher dimensional one.
Our main structural lemma for the discrete setting is the following:

\begin{lemma}[Correlation Lemma] \label{lem:cor-disc}
Let $k, m, M\in \Z_+$ with $k \leq m \leq M$.
If the distribution $A$ on $[m]\cup\{0\}$ satisfies Condition \ref{cond:moments-disc},
then for all $S,S' \subseteq [M]$ with $|S| = |S'| = m$, we have that
$|\chi_{U_M}(\p^A_S, \p^A_{S'})| \le (|S\cap S'|/m)^{k+1} \chi^2(A, \mathrm{Bin}(m,1/2)) + k\nu^2 \;.$
\end{lemma}
\begin{proof}
By definition, we have that
\begin{align*}
1+\chi^2(\p_S^A,U_M)
&=2^M \sum_{\bx\in\{0,1\}^M} \left(2^{-M+m} A\left(\sum_{i\in S}x_i \right)
\binom{m}{\sum_{i\in S}x_i}^{-1}\right)^2\\
&=2^m \sum_{j=0}^{m}A(j)^2/\binom{m}{j}=1+\chi^2(A,\bin(m,1/2)) \;,
\end{align*}
where we let $j$ denote $\sum_{i\in S}x_i$ in the second equality.

Now we proceed via discrete Fourier analysis.
Note that by Definition~\ref{def:char}, $\wh{\p^A_S}(T)=\E_{\bX\sim\p_S^A}[\chi_T(\bX)]$,
which is: (i) 0 if $T\not \subseteq S$,
(ii) $a_j/ \binom{m}{j}$ if $T\subseteq S$ and $|T|=j$,
where $a_j = \E_{t\sim A}[\mathcal{K}_j(t;m)]$.
This is because symmetry implies
that for each $T\subseteq S$ with $|T|=j$
we have that $\wh{\p^A_S}(T)$ is the same. Furthermore,
$\E_{t\sim A}[\mathcal{K}_j(t;m)] =
\E_{\bX\sim\p^A_S}\left[ \sum_{T\subseteq S, |T|=j} \chi_T(\bX) \right]$.
From this, by Fact~\ref{fact:chi}, we have that
\begin{align*}
1+\chi^2(A, \mathrm{Bin}(m,1/2)) &=
1+\chi^2(\p^A_S, U_M)= 1+\chi_{U_M}(\p_S^A,\p_S^A)
=\sum_{T\subseteq S}\wh{\p^A_S}(T)^2\\
&=\sum_{t=0}^m \sum_{T\subseteq S, |T|=t} \left(a_t/\binom{m}{t}\right)^2
= 1+\sum_{t=1}^m a_t^2 /\binom{m}{t} \;,
\end{align*}
where the last equality follows from the fact $\mathcal{K}_0(t;m)=1$.
In addition, by Fact~\ref{fact:chi}, we have that
\begin{align*}
1+\chi_{U_M}(\p^A_S, \p^A_{S'})
&=\sum_{T\subseteq S\cap S'}\wh{\p_S^A}(T)\wh{\p_{S'}^A}(T)
= \sum_{t=0}^m \left|\{T: |T|=t, T\subseteq S\cap S'  \}\right|a_t^2 / \binom{m}{t}^2 \\
&= 1 + \sum_{t=1}^{m} \binom{|S\cap S'|}{t} a_t^2 / \binom{m}{t}^2  \;,
\end{align*}
where the last equality follows from the fact $\mathcal{K}_0(t;m)=1$.
By Condition~\ref{cond:moments-disc}, we have that
\begin{align*}
\sum_{t=1}^k \binom{|S\cap S'|}{t}a_t^2 / \binom{m}{t}^2 \le \sum_{t=1}^k a_t^2\le k\nu^2 \;.
\end{align*}
The sum over terms with $t\geq k$ is at most
\begin{align*}
&\quad\sum_{t=k+1}^{m} \binom{|S\cap S'|}{t}a_t^2 / \binom{m}{t}^2
\le\left (\sum_{t=1}^{m} a_t^2 / \binom{m}{t}\right) \max_{t > k} \binom{|S\cap S'|}{t}/\binom{m}{t}\\
&=\left(\sum_{t=1}^{m} a_t^2 / \binom{m}{t}\right) \max_{t > k}\left(\frac{|S\cap S'|(|S\cap S'|-1)\cdots(|S\cap S'|-t+1)}{m(m-1)\cdots(m-t+1)}\right)\\
&\le\left(\sum_{t=1}^{m} a_t^2 / \binom{m}{t}\right)\left(\frac{|S\cap S'|}{m}\right)\left(\frac{|S\cap S'|-1}{m-1}\right)\cdots\left(\frac{|S\cap S'|-k}{m-k}\right)\\&\le \chi^2(A, \mathrm{Bin}(m,1/2))\left(|S\cap S'|/m\right)^{k+1} \;.
\end{align*}
This completes the proof.
\end{proof}

We will additionally require the following simple fact (see Appendix~\ref{ssec:hyper-geo} for the proof).

\begin{claim}\label{clm:near-orth-vec-disc}
Let $m,M\in\mathbb{Z}_+$ with $m<M$.
For any $0<c<1/2$ and $M>2m^{1+c}$,
there exists a collection $\mathcal{C}$ of $2^{m^{1-2c}/4}$
subsets $S\subseteq[M]$ with $|S|=m$ such that
any pair $S,S'\in\mathcal{C}$, with $S\ne S'$, satisfies $|S\cap S'|<m^{1-c}$.
\end{claim}

Combining the above, we obtain our generic discrete SQ hardness result:

\begin{proposition}[Generic Discrete SQ Hardness]\label{prop:gen-sq-prop}
Let $m,M\in\Z_+$ with $M>2m^{5/4}$.
Let $A$ be a distribution on $[m]\cup\{0\}$ satisfying
Condition~\ref{cond:moments-disc}. Let $\tau\ge m^{-(k+1)/4}\chi^2(A,\mathrm{Bin}(m,1/2))+k\nu^2$.
Any SQ algorithm that solves the testing problem of Definition~\ref{def:testing} with probability at least $2/3$
either makes queries of accuracy better than $\sqrt{2\tau}$
or makes at least $\frac{2^{\Omega(\sqrt{m})}\tau}{\chi^2(A,\mathrm{Bin}(m,1/2))}$ statistical queries.
\end{proposition}
\begin{proof}
Let $\mathcal{C}$ be a collection of $s=2^{\Omega(\sqrt{m})}$ subsets $S\subseteq[M]$
with $|S|=m$ whose pairwise intersections are all less than $m^{3/4}$.
By Claim~\ref{clm:near-orth-vec-disc} (taking the local parameter $c=1/4$),
such a set is guaranteed to exist. By Lemma~\ref{lem:cor-disc},
we have that for $S,S'\in\mathcal{C}$ with $S\ne S'$, it holds that
\begin{align*}
|\chi_{U_M}(\mathbf{P}_S^A,\mathbf{P}_{S'}^A)|\le m^{-(k+1)/4}\chi^2(A,\mathrm{Bin}(m,1/2))+k\nu^2\le\tau \;.
\end{align*}
If $S=S'$, then
$\chi_{U_M}(\mathbf{P}_S^A,\mathbf{P}_{S}^A)=\chi^2(\mathbf{P}_{S}^A,U_M)=\chi^2(A,\mathrm{Bin}(m,1/2))$.
Let $\gamma=\tau$ and $\beta=\chi^2(A,\mathrm{Bin}(m,1/2))$.
We have that the statistical query dimension of this testing problem
with correlations $\left(\gamma,\beta\right)$ is at least $s$.
Then applying Lemma~\ref{lem:sq-from-pairwise}
with $(\gamma,\beta)$ completes the proof.
\end{proof}

\section{SQ Lower Bound for Robustly Learning a Binary Product Distribution}\label{ssec:sq-bin}
In this section, we use the framework of Section~\ref{ssec:generic-discrete} to 
prove our super-polynomial SQ lower bound for robustly learning a binary product distribution.

Our SQ-hard instances for this case will be obtained using the product distributions
defined below.

\begin{definition}[SQ-Hard Instance for Binary Products]\label{def:bin-dist}
Let $0\le\delta\le1/2$ and $M,m\in\Z_+$ with $M>m$. 
For any subset $S\subseteq[M]$ with $|S|=m$, define $U_M^{S,\delta}$ to be 
the product distribution over $\{0,1\}^M$, where each coordinate has mean $1/2+\delta$ if it belongs to set $S$, 
and has mean $1/2$ otherwise. We let $\mu_{M}^{S,\delta}$ denote the mean vector of $U_{M}^{S,\delta}$.
\end{definition}

The following lemma states that the distributions in the above 
family are far from the uniform distribution $U_M$ in total variation distance. 
We defer the proof to Appendix~\ref{ssec:TV-binary-lower-bound}.

\begin{lemma}\label{lem:TV-binary-lower-bound}
Let $m,M\in\Z_+$ with $M>m$. Let $S\subseteq[M]$ with $|S|\le m$. 
Then for any sufficiently small $\delta>0$, $\dtv\Big(U_M,U_M^{S,\frac{\delta}{\sqrt{m}}}\Big)\ge\Omega(\delta)$, 
where $U_M$ is the uniform distribution over $\{0,1\}^M$.
\end{lemma}

The main result of this section is the following theorem:

\begin{theorem}[SQ Lower Bound for Robustly Testing a Binary Product Distribution]\label{thm:sq-bin-testing}
Fix $0<c<1/2$ and $k$ to be a sufficiently large integer. Let $m,M\in\Z_+$ with $M=3m^{5/4}$.
Let $0<\eps<1/2$ and $\delta$ be a sufficiently small constant multiple of $\epsilon\sqrt{\log(1/\epsilon)}/k^2$.
Let $\tau=\Theta(M^{-(k+1)/5}\delta)$.
Assume that $m>\max\Big(C'(\log(1/\epsilon))^{3},\frac{k^2}{\log(1/\delta)}\Big)$ for some sufficiently large constant $C'>0$.
Then any $\mathrm{SQ}$ algorithm which is given access to a distribution $\p$ over $\{0, 1\}^M$
so that either $H_0$: $\p=U_M$, 
or $H_1$: $\dtv\Big(\p,U_M^{S,\frac{\delta}{\sqrt{m}}}\Big)\le\epsilon$
for some unknown subset $S\subseteq[M]$ with $|S|=m$,
and correctly distinguishes between these two cases with probability at least $2/3$,
must either make queries of accuracy better than $\sqrt{2\tau}$
or must make at least
$2^{\Omega(M^{2/5})}M^{-(k+1)/5}$
statistical queries.
\end{theorem}

Theorem~\ref{thm:sq-bin-testing} will follow by applying our generic discrete SQ lower bound
construction in Section~\ref{ssec:generic-discrete} along with the following proposition.

\begin{proposition}\label{prop:bin-moment-matching}
Fix $\delta>0$ to be sufficiently small and $k$ to be a sufficiently large integer. 
Let integer $m\ge\max\big(C_0(\log(1/\delta))^3,\frac{k^2}{\log(1/\delta)}\big)$ 
for some universal constant $C_0>0$ sufficiently large. 
Then there exists a distribution $A$ over $[m]\cup\{0\}$ satisfying the following conditions: 
\begin{itemize}
\item[(i)] $\E_{X\sim A}[X^i]=\E_{X\sim\mathrm{Bin}(m,1/2)}[X^i]$ for all $1\le i\le k$, 
\item[(ii)] $\dtv(A,\mathrm{Bin}(m,1/2+\delta/\sqrt{m}))\le O\Big(\frac{\delta k^2}{\sqrt{\log(1/\delta)}}\Big)$, and 
\item[(iii)] $\chi^2(A,\mathrm{Bin}(m,1/2))=O(\delta)$.
\end{itemize}
\end{proposition}

In Section~\ref{ssec:moment-matching}, we give a technical overview of the proof. The detailed 
proof of Proposition~\ref{prop:bin-moment-matching} 
is deferred to Appendix~\ref{ssec:bin-moment-matching}.

\begin{proof}[Proof of Theorem~\ref{thm:sq-bin-testing}]
We can assume without loss of generality that $\epsilon>0$ is smaller than a sufficiently small universal constant.
Let
$\delta$ be a sufficiently small constant multiple of $\epsilon\sqrt{\log(1/\epsilon)}/k^2$.
From Proposition~\ref{prop:bin-moment-matching}, 
there is a distribution $A$ over $[m]\cup\{0\}$ such that (i) $A$ 
and $\mathrm{Bin}(m,1/2)$ agree on the first $k$ moments,
(ii) $\dtv(A,\mathrm{Bin}(m,1/2+\delta/\sqrt{m}))\le O(\delta k^2/\sqrt{\log(1/\delta)})$, and
(iii) $\chi^2(A,\mathrm{Bin}(m,1/2))=O(\delta)$.
In this way, for any subset $S\subseteq[M]$ with $|S|=m$, it holds that
\begin{align*}
&\quad\dtv\Big(\p_S^A,U_M^{S,\frac{\delta}{\sqrt{m}}}\Big)=
(1/2)\sum_{\bx\in\{0,1\}^M}\left|\p_S^A(\bx)-U_M^{S,\frac{\delta}{\sqrt{m}}}(\bx)\right|\\
&=(1/2)\sum_{\bx\in\{0,1\}^M}\Bigg|2^{-M+m} A\left(\sum_{i\in S}x_i\right)\binom{m}{\sum_{i\in S}x_i}^{-1}\\
&\quad-2^{-M+m}\left(1/2+\delta/\sqrt{m}\right)^{\sum_{i\in S}x_i}\left(1/2-\delta/\sqrt{m}\right)^{m-\sum_{i\in S}x_i}\Bigg|\\
&=(1/2)\sum_{j=0}^{m}\left|A(j)-\binom{m}{j}\left(1/2+\delta/\sqrt{m}\right)^{j}\left(1/2-\delta/\sqrt{m}\right)^{m-j}\right|\\
&=\dtv(A,\bin(m,1/2+\delta/\sqrt{m}))=O(\delta k^2/\sqrt{\log(1/\delta)})\leq \epsilon \;.
\end{align*}
By Claim~\ref{clm:near-orth-vec-disc},
there exists a collection $\mathcal{C}$ of $2^{\Omega(m)}$
subsets $S\subseteq[M]$ with $|S|=m$ such
that for any pair $S,S'\in\mathcal{C}$, with $S\ne S'$, satisfies $|S\cap S'|<m^{3/4}$.
Applying Proposition~\ref{prop:gen-sq-prop}, 
we determine that any SQ algorithm which, given access to a distribution $\p$
so that either $\p=U_M$, or $\p$ is given by $\p_S^A$
for some unknown subset $S\subseteq[M]$ with $|S|=m$,
correctly distinguish between these two cases with probability at least $2/3$
must either make queries of accuracy better than $\sqrt{2\tau}$
or must make at least
$\frac{2^{\Omega(\sqrt{m})}\tau}{\chi^2(A,\mathrm{Bin}(m,1/2))}\ge2^{\Omega(M^{2/5})}M^{-(k+1)/5}$
statistical queries, since $m^{-(k+1)/4}\chi^2(A,\mathrm{Bin}(m,1/2))\le O(M^{-(k+1)/5}\delta)\le\tau$.
This completes the proof of Theorem~\ref{thm:sq-bin-testing}.
\end{proof}

\section{SQ Lower Bound for Robustly Learning a Ferromagnetic High-Temperature Ising Model}\label{ssec:sq-ising}
In this section, we
prove our super-polynomial SQ lower bound for robustly learning a ferromagnetic high-temperature Ising model. 
We start by transforming the support of Ising models to $\{0,1\}^M$:

\begin{definition}\label{def:sq-ising}
Given a real symmetric matrix $(\theta_{ij})_{i,j\in[M]}$ with zero diagonal,
the Ising model distribution $P_\theta$ is defined as follows: For any $\bx\in\{0,1\}^M$,
$P_\theta(\bx)=\frac{1}{Z(\theta)}\exp \big((1/2)\sum_{i,j\in[M]}{(-1)^{x_i+x_j}\theta_{ij}}\big)$,
where the normalizing factor $Z(\theta)$ is called the partition function.
We call the matrix $(\theta_{ij})_{i,j\in[M]} \in \R^{M \times M}$ the interaction matrix.
\end{definition}

Our SQ-hard instances for this case will be obtained using the Ising models
defined below.

\begin{definition}[SQ-Hard Instances for Ising Models]\label{def:ising}
Let $m,M\in\Z_+$ with $M>m$.
Let $0\le\delta\le\frac{1}{2m}$. For every subset $S\subseteq[M]$ with $|S|=m$,
define $Q_{M}^{S,\delta}$ to be the Ising model
with parameter $\theta$, where for every pair $i\ne j\in[M]$
we have that $\theta_{ij}=\delta,\forall i,j\in S$ and $\theta_{ij}=0$ otherwise.
Note that by our choice of parameter $\delta$,
the Ising models $Q_{M}^{S,\delta}$ are both
high-temperature and ferromagnetic.
\end{definition}

The following lemma states that the distributions in the above family
are far from the uniform distribution $U_M$ in total variation distance.
We defer the proof to Appendix~\ref{ssec:TV-ising-lower-bound}.

\begin{lemma}\label{lem:TV-ising-lower-bound}
Let $m,M\in\Z_+$ with $M>m$. Let $S\subseteq[M]$ with $|S|\le m$.
Then, for any sufficiently small $\delta>0$, we have that
$\dtv\Big(U_M, Q_M^{S,\frac{\delta}{m}}\Big)\ge\Omega(\delta)$.
\end{lemma}

The main result of this section is the following theorem:

\begin{theorem}[SQ Lower Bound for Robustly Testing Ising Models]\label{thm:sq-ising-testing}
Fix $0<c<1$ and $k$ to be a sufficiently large integer. 
Let $m,M\in\Z_+$ with $M=3m^{5/4}$. Let $0<\eps<1/2$
and $\delta$ be a sufficiently small multiple of $\epsilon\log(1/\epsilon)/k^3$. 
Let $\tau=\Theta(M^{-(k+1)/5}\delta)$. 
Assume that $m>\max\Big(C'(\log(1/\epsilon))^{3},\frac{k^2}{\log(1/\delta)}\Big)$ 
for some sufficiently large constant $C'>0$.
Then any $\mathrm{SQ}$ algorithm which is given access to a distribution $\p$ over $\{0,1\}^M$
so that either $H_0$: $\p=U_M$, or $H_1$: $\dtv\Big(\p,Q_M^{S,\frac{\delta}{m}}\Big)\le\epsilon$
for some unknown subset $S\subseteq[M]$ with $|S|=m$,
and correctly distinguishes between these two cases with probability at least $2/3$
must either make queries of accuracy better than $\sqrt{2\tau}$ or
must make at least $2^{\Omega(M^{2/5})}M^{-(k+1)/5}$ statistical queries.
\end{theorem}

We will additionally require the following definition.

\begin{definition}\label{def:ising-sym}
Fix $n$ to be a positive integer.
Let $\is(n,\delta)$ be the distribution over $[n]\cup\{0\}$ with 
$\is(n,\delta)(x)=\binom{n}{x}\exp\left(h(n,x)\delta\right)/Z_n(\delta)$ 
for some parameter $-1/n<\delta<1/n$, where $h(n,x)=2x^2-2nx+\frac{n(n-1)}{2}$ 
and $Z_n(\delta)=\sum_{x=0}^{n}\binom{n}{x}\exp(h(n,x)\delta)$.
\end{definition}

By Definition~\ref{def:sq-ising}, $Z_n(\delta)$ is the partition function of the Ising model over $\{0,1\}^n$, 
where every entry outside of the diagonal of the interaction matrix is $\delta$. 
Intuitively, $\is(n,\delta)(x)$ denotes the contribution of the configurations containing $x$ 1's in the Ising model.

Theorem~\ref{thm:sq-ising-testing} will follow by applying our generic discrete SQ lower bound
construction of Section~\ref{ssec:generic-discrete} along with the following proposition. 

\begin{proposition}\label{prop:ising-moment-matching}
Fix $\delta>0$ to be sufficiently small and $k$ to be an arbitrary positive integer. 
Let integer $m\ge \max\big(C_0(\log(1/\delta))^3,\frac{k^2}{\log(1/\delta)}\big)$ 
for some universal constant $C_0>0$ sufficiently large. 
Then there exists a distribution $A$ over $[m]\cup\{0\}$ satisfying the following conditions: 
\begin{itemize}
\item[(i)] $\E_{X\sim A}[X^i]=\E_{X\sim\mathrm{Bin}(m,1/2)}[X^i]$ for all $1\le i\le k$, 
\item[(ii)] $\dtv(A,\is(m,\delta/m))\le O\big(\frac{\delta k^3}{\log(1/\delta)}\big)$, and 
\item[(iii)] $\chi^2(A,\mathrm{Bin}(m,1/2))=O(\delta)$.
\end{itemize}
\end{proposition}

In Section~\ref{ssec:moment-matching}, we give a technical overview of the proof. The detailed 
proof of Proposition~\ref{prop:ising-moment-matching} 
is deferred to Appendix~\ref{ssec:ising-moment-matching}.

\begin{proof}[Proof of Theorem~\ref{thm:sq-ising-testing}]
We can assume without loss of generality that $\epsilon>0$
is smaller than a sufficiently small universal constant. 
Let $\delta$ be a sufficiently small constant multiple of $\epsilon\log(1/\epsilon)/k^3$.
From Proposition~\ref{prop:ising-moment-matching}, there is a distribution $A$ over $[m]\cup\{0\}$ 
such that (i) $\E_{X\sim A}[X^i]=\E_{X\sim\mathrm{Bin}(m,1/2)}[X^i]$ for all $0\le i\le k$, 
(ii) $\dtv(A,\is(m,\delta/m))\le O\left(\frac{\delta k^3}{\log(1/\delta)}\right)\le O(\epsilon)$, 
and (iii) $\chi^2(A,\mathrm{Bin}(m,1/2))=O(\delta)$.
Note that for any subset $S\subseteq[M]$ with $|S|=m$, it holds that
\begin{align*}
&\quad\dtv\Big(\p_S^A,Q_M^{S,\frac{\delta}{m}}\Big)=
(1/2)\sum_{\bx\in\{0,1\}^M}\left|\p_S^A(\bx)-Q_M^{S,\frac{\delta}{m}}(\bx)\right|\\
&=(1/2)\sum_{\bx\in\{0,1\}^M}\left|2^{-M+m} A \left(\sum_{i\in S}x_i\right)\binom{m}{\sum_{i\in S}x_i}^{-1}- 2^{-M+m}\left(\frac{\exp(h(m,\sum_{i\in S}x_i)\delta/m)}{Z_n(\delta/m)}\right)\right|\\
&=(1/2)\sum_{j=0}^{m}\left|A(j)-\frac{\binom{m}{j} \exp(h(m,j)\delta/m)}{Z_n(\delta/m)}\right| 
=\dtv(A,\is(m,\delta/m))=O(\delta k^3/\log(1/\delta)) \le \epsilon \;.
\end{align*}
By Claim~\ref{clm:near-orth-vec-disc}, there exists a collection $\mathcal{C}$
of $2^{\Omega(m)}$ subsets $S\subseteq[M]$ with $|S|=m$
such that for any pair $S,S'\in\mathcal{C}$, with $S\ne S'$, satisfies $|S\cap S'|<m^{3/4}$.
Applying Proposition~\ref{prop:gen-sq-prop}, we determine
that any SQ algorithm which, given access to a distribution $\p$
so that either $\p=U_M$, or $\p$ is given by $\p_S^A$
for some unknown subset $S\subseteq[M]$ with $|S|=m$,
correctly distinguishes between these two cases with probability at least $2/3$
must either make queries of accuracy better than $\sqrt{2\tau}$
or must make at least
$\frac{2^{\Omega(\sqrt{m})}\tau}{\chi^2(A,\mathrm{Bin}(m,1/2))}\ge2^{\Omega(M^{2/5})}M^{-(k+1)/5}$
statistical queries, since $m^{-(k+1)/4}\chi^2(A,\mathrm{Bin}(m,1/2))\le O(M^{-(k+1)/5}\delta)\le\tau$.
This completes the proof of Theorem~\ref{thm:sq-ising-testing}.
\end{proof}

\section{Proof Sketch of Proposition~\ref{prop:bin-moment-matching} 
and Proposition~\ref{prop:ising-moment-matching}}\label{ssec:moment-matching}

The construction of the distribution $A$ in both cases
is similar in spirit to the technique in~\cite{DKS17-sq} for constructing
a distribution that matches moments with $\mathcal{N}(0,1)$
but is close in total variation distance to $\mathcal{N}(\delta,1)$, for an appropriate $\delta>0$.
Specifically, we start from some appropriate one-dimensional version 
of either a binary product distribution or Ising model, $H(x)$, 
over $[m]\cup\{0\}$, and then modify it in order to match the first $k$ moments 
with $\bin(m,1/2)$. We achieve this by modifying
the probability mass function of $A$ by adding a polynomial $q$
over some appropriately chosen interval $I=[(1/2-C)m,(1/2+C)m]$, for some carefully selected 
$C=\Theta(\sqrt{(\log(1/\delta)/m)}$.
In particular, for any integer point $x\in I$, 
we let $q(x)=\int_{x}^{x+1}p(t)dt$, for some real polynomial $p$ of degree-$k$ and then modify the probability mass function by adding $q(x)$ to $H(x)$.
The moment-matching condition 
amounts to a system of linear equations on the coefficients of $p$.
We show that this system has a unique solution.
Then the rest of our analysis focuses on showing that this modification
leaves the probability mass function of $A$ still non-negative
and sufficiently close to $H$ in total variation distance.

In particular, we express the polynomial $p$ as a linear combination 
of appropriately scaled Legendre polynomials, 
i.e., $p(t)=\sum_{i=0}^{k}a_iP_i\left(\frac{t-m/2}{Cm}\right)$,
where $P_i$ denotes the $i$-th Legendre polynomial and $a_i\in\R$ is a coefficient. 
Then we show, by analogy to the proof in~\cite{DKS17-sq}, 
that the $L_1$ and $L_\infty$ norms of $p$ within the interval $I$ are sufficiently small.
In particular, the~\cite{DKS17-sq} result on the hardness of robustly 
learning unknown-mean or covariance Gaussians 
essentially solves
the limiting version of this problem (that is achieved as $m\to\infty$).
As their analysis shows that this limiting case works,
we need to show that when $m$ is sufficiently large,
we are sufficiently close to that limiting case that our construction will also succeed.
To achieve this, we require some new proof ideas 
in order to show that with sufficiently large but finite $m$, 
our analysis will be close enough to that of the limiting case, 
so that the results of~\cite{DKS17-sq} can still be applied.

In more detail, by our construction of the polynomial $p$ and the moment-matching condition, 
we are able to bound from above the coefficients $a_i$ as follows: 
\begin{align*}
|a_i|=\left(\frac{2i+1}{2Cm}\right)\left|\int_{t\in I}p(t)P_i\left(\frac{t-m/2}{Cm}\right)dt\right|\le\frac{(2i+1)(\gamma_i+\beta_i)}{2Cm}, 
&\qquad\forall 1\le i\le k \;,
\end{align*}
where we have that $\beta_i=\left|\sum_{x=0}^{m}(H(x)-\bin(m,1/2)(x))P_i\left(\frac{x-m/2}{Cm}\right)\right|$ and
\begin{align*}
\gamma_i=\left|\sum_{x\in\Z\cap I}P_i\left(\frac{x-m/2}{Cm}\right)\int_{x}^{x+1}p(t)dt-\int_{t\in I}p(t)P_i\left(\frac{t-m/2}{Cm}\right)dt\right| \;.
\end{align*}
Intuitively speaking, the quantity $\beta_i$ represents the answer to the continuous version of the problem and the quantity $\gamma_i$ inherently captures 
the error between the discrete and limiting continuous versions of our problems.
Since the absolute value of the derivative of the $i$-th Legendre polynomial 
is at most $O(i^2)$ in the interval $[-1,1]$, by the last property of Fact~\ref{fact:legendre-poly},
we are able to apply the mean-value theorem to obtain an upper bound for 
$\gamma_i$ in terms of the $L_1$-norm of $p$ within the interval $I$.
For the quantity $\beta_i$, we borrow ideas from~\cite{DKS17-sq} 
to view $H(x)-\bin(m,1/2)(x)$ as a function of some appropriately 
chosen parameter, and apply Taylor's theorem to expand this difference up to second order terms. 
Then we can show that both the first order and second order terms are sufficiently small. 

In summary, we prove the following technical result.

\begin{theorem}\label{thm:moment-matching}
Fix $\delta>0$, $0<C<1/2$ and $k\in\Z_+$. 
Let integer $m>k/(2C)$ such that both $(1/2-C)m$ and $(1/2+C)m$ are integers. 
Consider the interval $I_{C,m}=[(1/2-C)m,(1/2+C)m-1]$. 
Let $\{H_{m,x}(t)\}_{x\in[m]\cup\{0\}}$ be a family of real functions. 
Then there is a unique real polynomial $p$ of degree at most $k$ such that
\begin{align}\label{eq:moment-matching}
\sum_{x\in\Z\cap I_{C,m}}x^i\int_{x}^{x+1}p(t)dt=\sum_{x=0}^{m}(H_{m,x}(0)-H_{m,x}(\delta))x^i:=b_i, 
&\qquad\forall 0\le i\le k \;.
\end{align}
In addition, we can write $p(t)=\sum_{i=1}^{k}a_iP_i\left(\frac{t-m/2}{Cm}\right)$, where
\begin{align}\label{eq:bound-a}
|a_i|\le\left(\frac{2i+1}{2Cm}\right)\left(\beta_i+O\left(\frac{i^2}{Cm}\right)\int_{(1/2-C)m}^{(1/2+C)m}|p(t)|dt\right) \;,
\end{align}
for all $1\le i\le k$, 
where $\beta_i=\left|\sum_{x=0}^{m}(H_{m,x}(0)-H_{m,x}(\delta))P_i\left(\frac{x-m/2}{Cm}\right)\right|$.
\end{theorem}

\begin{proof}
We first show that there is a unique real polynomial $p$ 
of degree at most $k$ satisfying~\eqref{eq:moment-matching}.
Let $p(t)=\sum_{i=0}^{k}p_it^i$ and $q(x)=\int_{x}^{x+1}p(t)dt$.
We note that each value of $i$ implies a single linear condition on $q$. 
This suggests that as long as the support domain $I_{C,m}$ is sufficiently large,
we can simply solve a system of linear equations to find it.
In more detail, we start by establishing the relationship between $p$ and $q$.
By definition, we have that
\begin{align*}
q(x)
&=\sum_{i=0}^{k}\frac{p_i((x+1)^{i+1}-x^{i+1})}{i+1} 
=\sum_{i=0}^{k}\left(\frac{p_i}{i+1}\right)\sum_{j=0}^{i}\binom{i+1}{j}x^j\\
&=\sum_{j=0}^{k}x^j\sum_{i=j}^{k}\left(\frac{p_i}{i+1}\right)\binom{i+1}{j}=\sum_{j=0}^{k}q_jx^j \;,
\end{align*}
where
\begin{align}\label{eq:p-q}
q_j=\sum_{i=j}^{k}\left(\frac{p_i}{i+1}\right)\binom{i+1}{j} \;.
\end{align}
This gives us a linear equation to solve for $p_i$ in terms of $q_j$ that is upper triangular and thus has a unique solution.
For any two polynomials $r_1(x),r_2(x)$ of degree at most $k$, 
we consider the inner product $\langle r_1,r_2\rangle\in\R$ given by the following:
\begin{align*}
\langle r_1,r_2\rangle:=\sum_{x\in\Z\cap I_{C,m}}r_1(x)r_2(x) \;.
\end{align*}
To show that this inner product is non-degenerate as long as $m$ is sufficiently large,
we need to show that for any polynomial $r(x)$ of degree at most $k$, 
it holds that $\sum_{x\in\Z\cap I_{C,m}}r^2(x)=0\Rightarrow r=0$. 
By our assumptions of $C,k,m$, $\sum_{x\in\Z\cap I_{C,m}}r^2(x)=0$ 
will imply that the polynomial $r(x)$ of degree at most $k$ 
has at least $2Cm>k$ different roots, which implies $r=0$.
Therefore, we can write the LHS of~\eqref{eq:moment-matching} as
\begin{align}\label{eq:coefficient}
\langle x^i,q(x)\rangle=b_i,\qquad0\le i\le k \;.
\end{align}
Since $1,x,\cdots,x^k$ are linearly independent polynomials of degree at most $k$, 
there exists a unique polynomial $q$ of degree at most $k$ 
satisfying the system of equations~\eqref{eq:coefficient}.

To show inequality~\eqref{eq:bound-a}, we first express $p$ 
as a linear combination of scaled Legendre polynomials 
whose coefficients are explicitly given by integrals. 
In particular, since $p$ has degree at most $k$ and the set of polynomials 
$\left\{P_i\left(\frac{t-m/2}{Cm}\right)\right\}_{0\le i\le k}$ 
contains a polynomial of each degree from $0$ to $k$, 
there exist $a_i\in\R$ such that $p(t)=\sum_{i=0}^{k}a_iP_i\left(\frac{t-m/2}{Cm}\right)$. 
It follows from Fact~\ref{fact:legendre-poly} (ii) that
\begin{align*}
\int_{(1/2-C)m}^{(1/2+C)m}p(t)P_i\left(\frac{t-m/2}{Cm}\right)dt 
&=\sum_{j=0}^{k}a_j\int_{(1/2-C)m}^{(1/2+C)m}P_i\left(\frac{t-m/2}{Cm}\right)P_j\left(\frac{t-m/2}{Cm}\right)dt\\
&=Cm\sum_{j=0}^{k}a_j\int_{-1}^{1}P_i(t)P_j(t)dt=\frac{2Cma_i}{2i+1} \;,
\end{align*}
which implies that $a_i=\left(\frac{2i+1}{2Cm}\right)\int_{(1/2-C)m}^{(1/2+C)m}p(t)P_i\left(\frac{t-m/2}{Cm}\right)dt$, 
for all $0\le i\le k$. In addition, by Equation~\eqref{eq:moment-matching}, 
we have that $a_0=\frac{1}{2Cm}\int_{(1/2-C)m}^{(1/2+C)m}p(t)dt=0$.
We now bound $|a_i|$ as follows. Let
$$\gamma_i=
\Bigg|\sum_{x\in\Z\cap I_{C,m}}P_i\left(\frac{x-m/2}{Cm}\right)\int_{x}^{x+1}p(t)dt-\int_{(1/2-C)m}^{(1/2+C)m}p(t)P_i\left(\frac{t-m/2}{Cm}\right)dt\Bigg| \;.$$ 
By Fact~\ref{fact:legendre-poly} (viii) and the mean-value theorem, we have that
\begin{align*}
\gamma_i
&=\Bigg|\sum_{x\in\Z\cap I_{C,m}}\int_{x}^{x+1}\left(P_i\left(\frac{x-m/2}{Cm}\right)-P_i\left(\frac{t-m/2}{Cm}\right)\right)p(t)dt\Bigg|\\
&=\left(\frac{1}{Cm}\right)\Bigg|\sum_{x\in\Z\cap I_{C,m}}\int_{x}^{x+1}(x-t)P_i'\left(\frac{\xi_t-m/2}{Cm}\right)p(t)dt\Bigg|\\
&\le O\left(\frac{i^2}{Cm}\right)\sum_{x\in\Z\cap I_{C,m}}\int_{x}^{x+1}|p(t)|dt 
=O\left(\frac{i^2}{Cm}\right)\int_{(1/2-C)m}^{(1/2+C)m}|p(t)|dt \;,
\end{align*}
where $\xi_t$ is some real number between $x$ and $t$ for each $t\in[x,x+1)$. 
Therefore, by Equation~\eqref{eq:moment-matching}, we have that
\begin{align*}
|a_i|
&=\left|\left(\frac{2i+1}{2Cm}\right)\int_{(1/2-C)m}^{(1/2+C)m}p(t)P_i\left(\frac{t-m/2}{Cm}\right)dt\right|\\
&\le\left(\frac{2i+1}{2Cm}\right)\Bigg(\gamma_i+\Bigg|\sum_{x\in\Z\cap I_{C,m}}P_i\left(\frac{x-m/2}{Cm}\right)\int_{x}^{x+1}p(t)dt\Bigg|\Bigg)\\
&=\left(\frac{2i+1}{2Cm}\right)(\gamma_i+\beta_i) 
\le \left(\frac{2i+1}{2Cm}\right)\left(\beta_i+O\left(\frac{i^2}{Cm}\right)\int_{(1/2-C)m}^{(1/2+C)m}|p(t)|dt\right) \;,
\end{align*}
where the last equality follows from Equation~\eqref{eq:moment-matching} and
$$\sum_{x\in\Z\cap I_{C,m}}P_i\left(\frac{x-m/2}{Cm}\right)\int_{x}^{x+1}p(t)dt=\sum_{x=0}^{m}(H_{m,x}(0)-H_{m,x}(\delta))P_i\left(\frac{x-m/2}{Cm}\right),$$
since $P_i\left(\frac{x-m/2}{Cm}\right)$ is a polynomial in $x$ of degree $i$.
This completes the proof.
\end{proof}


\bibliographystyle{alpha}

\bibliography{allrefs}

\newpage

\appendix

\section*{Appendix}

\section{Omitted Technical Preliminaries}\label{app:prelims-1}

In this section, we record the required definitions and technical facts.

\subsection{Basics Facts }\label{ssec:binomial-expectation}

\begin{fact}\label{fact:bin-bound-1}
$\binom{2n}{n}\le\frac{2^{2n}}{\sqrt{2n}},\forall n\in\Z_{+}\setminus\{0\}$ and $\binom{2n+1}{n}\le\frac{2^{2n+1}}{\sqrt{2n+1}},\forall n\in\Z_+$.
\end{fact}

\begin{fact}[\cite{CoverThomas:91}]\label{fact:bin-bound-2}
Let $n,k\in\Z$. Then, we have that
\begin{align*}
\sqrt{\frac{n}{8k(n-k)}}2^{nH(k/n)}\le\binom{n}{k}\le\sqrt{\frac{n}{\pi k(n-k)}}2^{nH(k/n)},
\end{align*}
where $H(p)=-p\log p-(1-p)\log(1-p)$ is the binary entropy function.
\end{fact}
We will use the following fact to bound from below the expectation of a real random variable.
\begin{fact}\label{fact:E-lower}
Let $X$ be a real random variable with $\E[X^4]>0$. Then, we have that $\E[|X|]\ge\frac{\E[X^2]^{3/2}}{\E[X^4]^{1/2}}$.
\end{fact}

\subsection{Sub-Gaussian and Sub-Exponential Distributions}\label{ssec:Gaussian-exp} 
Here we present basic facts about sub-Gaussian and sub-exponential distributions.
The reader is referred to~\cite{vershynin2018high}.
\begin{definition}[Sub-Gaussian Distribution]\label{def:Gaussian-decay}
A random variable $X$ over $\R$ is sub-Gaussian if $\|X\|_{\psi_2}:=\inf\{t>0:\E[\exp(X^2/t^2)]\le2\}$ is finite.
\end{definition}

\begin{definition}[Sub-Exponential Distribution]\label{def:exp-decay}
A random variable $X$ over $\R$ is sub-exponential
if $\|X\|_{\psi_1}:=\inf\{t>0:\E[\exp(|X|/t)]\le2\}$ is finite.
\end{definition}

\begin{fact}\label{fact:Gaussian-decay}
Let $X$ be a real random variable. Suppose there is a real number $K>0$ such that $\Pr[|X|>t]\le2\exp(-t^2/K^2)$.
Then $X$ is sub-Gaussian with $\|X\|_{\psi_2}\le cK$ for some universal constant $c>0$.
In addition, we have that $\E[|X|^p]\le\min\left(pK^p\left\lfloor\frac{p-1}{2}\right\rfloor!,(c'K\sqrt{p})^p\right)$, where $c'>0$ is a universal constant.
\end{fact}

\begin{fact}\label{fact:exp-decay}
Let $X$ be a real random variable. Suppose there is a real number $K>0$ such that $\Pr[|X|>t]\le2\exp(-t/K)$.
Then $X$ is sub-exponential with $\|X\|_{\psi_1}\le cK$ for some universal constant $c>0$.
In addition, we have that $\E[|X|^p]\le2K^pp!\le2(Kp)^p$.
\end{fact}

\begin{fact}\label{fact:Gaussian-exp-1}
$\|\cdot\|_{\psi_2}$ is a norm on the space of sub-Gaussian random variables.
$\|\cdot\|_{\psi_1}$ is a norm on the space of sub-exponential random variables.
\end{fact}

\begin{fact}\label{fact:center-Gaussian-exp}
Let $X$ be sub-Gaussian and $Y$ be sub-exponential with $\E[X]=\E[Y]=0$.
Then there exists universal constants $C_1,C_2>0$ such that $\E[\exp(\lambda X)]\le\exp(C_1^2\lambda^2\|X\|_{\psi_2}^2), \forall\lambda\in\R$ and $\E[\exp(\lambda Y)]\le\exp(C_2^2\lambda^2\|Y\|_{\psi_1}^2) ,|\lambda|\le\frac{1}{C_2\|Y\|_{\psi_1}}$.
\end{fact}

\begin{fact}\label{fact:Gaussian-exp-2}
Let $X$ be sub-Gaussian and $Y$ be sub-exponential.
Then, there exists universal constants $c_1,c_2>0$ such that $\|X-\E[X]\|_{\psi_2}\le c_1\|X\|_{\psi_2}$ and $\|Y-\E[Y]\|_{\psi_1}\le c_2\|Y\|_{\psi_1}$.
\end{fact}

\subsection{Dobrushin's Uniqueness Condition} \label{app:dob}

Here we introduce the original definition of Dobrushin's condition 
through the influence between points in general graphical model.

\begin{definition}[Influence in Graphical Models]\label{Dobrushin's-condition}
Let $D$ be a distribution over some set of points $V$. 
Let $S_j$ denote the set of state pairs $(X,Y)$ which differ only at point $j$. 
Then the influence of point $j\in V$ on point $i\in V$ is defined as
\begin{align*}
I(j,i)=\max_{(X,Y)\in S_j}\dtv(D_i(\cdot\mid X_{-i}),D_i(\cdot\mid Y_{-i})) \;,
\end{align*}
where $D_i(\cdot\mid X_{-i}),D_i(\cdot\mid Y_{-i})$ 
denote the marginal distribution of point $i$ conditioning on $X_{-i}$ and $Y_{-i}$ respectively.
\end{definition}

\begin{definition}[Dobrushin's Uniqueness Condition]
Let $D$ be a distribution over some set of points $V$.
Then $D$ is said to satisfy Dobrushin's uniqueness condition if $\max_{i\in V}\sum_{j\in V}{I(j,i)}<1$.
\end{definition}

For Ising models,~\cite{chatterjee2005concentration} proves that 
$\max_{i\in V}\sum_{j\ne i}{|\theta_{ij}|}<1$ implies the Dobrushin's uniqueness condition.

\subsection{Concentration of Ising Models}

Several recent works have studied the concentration and anti-concentration 
of functions of Ising models~\cite{gheissari2018concentration,gotze2019higher, daskalakis2017concentration, adamczak2019note}.
Here we record some results which will be used throughout this article.

The following two facts state that for an Ising model satisfying Dobrushin's condition, for some constant $\eta>0$, 
the linear form and the quadratic form of Ising models are sub-Gaussian and sub-exponential respectively.
\begin{fact}[\cite{gotze2019higher}]\label{fact:ising-sub-Gaussian}
Let $P_{\theta}$ be an Ising model satisfying Dobrushin's condition, i.e., $\max_{i\in[d]}{\sum_{j\ne i}|\theta_{ij}|}\le1-\eta$, for some constant $0<\eta<1$.  
Then there is a constant $c(\eta)>0$ such that for any $b\in \R^d$ and any $t>0$, we have that
$\Pr_{X\sim P_\theta} \left[ \left| b^TX - \E_{X\sim P_\theta}\left[b^TX\right] \right| > t \right] 
\le 2 \exp\Big(-\frac{t^2}{c(\eta)\|b\|_2^2}\Big)$. This implies that $\left\|b^TX - \E\left[b^TX\right]\right\|_{\psi_2}\le c'(\eta)\|b\|_2$ for some constant $c'(\eta)>0$.
\end{fact}

\begin{fact}[\cite{gotze2019higher}]\label{fact:ising-sub-exponential}
Let $P_{\theta}$ be an Ising model satisfying Dobrushin's condition, i.e., $\max_{i\in[d]}{\sum_{j\ne i}|\theta_{ij}|}\le1-\eta$, for some constant $0<\eta<1$.  
Then there is a constant $c(\eta)>0$ such that for any symmetric matrix $A\in\R^{d\times d}$ with zero diagonal and any $t>0$, we have that
$\Pr_{X\sim P_\theta} \left[ \left|X^TAX - \E_{X\sim P_\theta}\left[X^TAX\right] \right| > t \right] 
\le 2 \exp\left(-\frac{t}{c(\eta)\|A\|_F}\right)$. This implies that \\$\left\|X^TAX - \E_{X\sim P_\theta}\left[X^TAX\right]\right\|_{\psi_1}\le c'(\eta)\|A\|_F$ for some constant $c'(\eta)>0$.
\end{fact}

\subsection{Basic Facts about the Hypergeometric Distribution}\label{ssec:hyper-geo}
Let $k,n,N\in\Z_+$.
Consider an urn consisting of $N$ balls in total among which $k$ are red, and $N-k$ are blue.
Let $X$ denote the number of red balls obtained by sampling $n$ balls from the urn \emph{without} replacements.
In this way, we say that $X\sim\mathrm{Hypergeom}(K,N,n)$.
We will also use the following standard fact:
\begin{fact}\label{fact:hyper-hoeffding}
Let $X\sim\mathrm{Hypergeom}(K,N,n)$ and $p=K/N$. Then for any $t>0$, we have that
\begin{align*}
\Pr\left[X>np+t\right]\le\exp\left(-2t^2/n\right).
\end{align*}
\end{fact}

\subsubsection{Proof of Claim~\ref{clm:near-orth-vec-disc}}
Let $S$ and $S'$ be independent uniformly random subsets from $[M]$ with $|S|=|S'|=m$.
Note that $|S\cap S'|\sim\mathrm{Hypergeom}(m,M,m)$, by Fact~\ref{fact:hyper-hoeffding}, we know that
\begin{align*}
\Pr[|S\cap S'|\ge m^{1-c}]\le\Pr\left[|S\cap S'|\ge m\left(\frac{m}{M}+\frac{m^{-c}}{2}\right)\right]\le\exp\left(-\frac{m^{1-2c}}{2}\right).
\end{align*}
Therefore, by the union bound,
\begin{align*}
\Pr[\exists |S|=|S'|=m:|S\cap S'|\ge m^{1-c}]\le2^{\frac{m^{1-2c}}{2}}\cdot\exp\left(-\frac{m^{1-2c}}{2}\right)<1.
\end{align*}

\subsection{Reduction of Testing to Learning}\label{ssec:learning-to-testing}

We have the following simple claim:

\begin{claim}\label{clm:learning-to-testing}
Suppose there exists an SQ algorithm to learn an unknown distribution in $\D$ to total variation distance $\eps$ using at most $N$ statistical queries of tolerance $\tau$.
Suppose furthermore that for each $D'\in\D$ we have that $\dtv(D,D')>2(\tau+\eps)$.
Then there exists an SQ algorithm that solves the testing problem $\mathcal{B}(\D,D)$ using at most $n+1$ queries of tolerance $\tau$.
\end{claim}
\begin{proof}
We begin by running the learning algorithm under the assumption that the unknown distribution in question is $D_0\in\mathcal{D}$ to get a hypothesis distribution $D'$.
We let $S$ be a subset so that $\dtv(D,D')=|D(S)-D'(S)|$, and use an additional statistical query to get an estimate $v$ of the expectation of $\mathbb{I}[S]$, the indicator function of $S$.
If the original distribution was $D$, we have that $|v-D(S)|\le\tau$. If the original distribution was $D_0$, we have that $|v-D'(S)|\le|v-D_0(S)|+|D_0(S)-D'(S)|\le\tau+\epsilon$.
However, we have that
\begin{align*}
|D(S)-D'(S)|=\dtv(D,D')\ge\dtv(D,D_0)-\dtv(D_0,D')>2(\tau+\epsilon)-\epsilon=2\tau+\epsilon \;.
\end{align*}
Therefore, our distribution is in $\D$ if and only if the expectation of $\mathbb{I}[S]$ is within $\tau+\epsilon$ of $D'(S)$.
Thus, determining which of these cases holds will solve our decision problem.
\end{proof}

\section{Omitted Statements and Proofs from Section~\ref{ssec:sq-bin}}\label{app:sq-bin}
\begin{definition}\label{def:balance}
Fix $0<c<1/2$ to be a constant. We say that a binary product distribution is $c$-balanced if every coordinate of the mean vector is in $[c,1-c]$.
\end{definition}
For $c$-balanced binary product distributions, we have the following lemma.
\begin{lemma}\label{lem:tv-balanced-bin}
Let $P$ and $Q$ be $c$-balanced binary product distributions with mean vectors $\mu_p$ and $\mu_q$. Then, $\dtv(P,Q)\le O(\|\mu_p-\mu_q\|_2/\sqrt{c})$.
\end{lemma}

We provide the result for hardness of robust learning of an unknown binary product distribution here.
In order to make the distributions in our family far from the reference distribution $U_M$ in total variation distance, we need higher dimension $m,M$ compared with the hardness result for robust hypothesis testing.
\begin{theorem}[SQ Lower Bound for Robust Learning of a Binary Product Distribution]\label{thm:sq-bin-tv}
Fix $0<c<1$ and $k$ to be a sufficiently large integer. Let $m,M\in\Z_+$ with $M=3m^{5/4}$. Let $0<\eps<1/2$ and $\delta$ be a sufficiently small multiple of $\epsilon\sqrt{\log(1/\epsilon)}/k^2$.
Let $\tau=\Theta(M^{-(k+1)/5}\delta)$.
Assume that $m>\max\Big(C'/\epsilon,\frac{k^2}{\log(1/\delta)}\Big)$ for some sufficiently large constant $C'>0$.
Then any $\mathrm{SQ}$ algorithm which is given access to a distribution $\p$ over $\{0,1\}^M$ which satisfies $\dtv\Big(\p,U_M^{S,\frac{\delta}{\sqrt{m}}}\Big)\le\epsilon$
for some unknown subset $S\subseteq[M]$ with $|S|=m$, outputs a hypothesis $\q$ with $\dtv(\q,\p)\le O(\delta)$ with probability at least $2/3$ must either make queries of accuracy better than $\sqrt{2\tau}$ or
must make at least $2^{\Omega(M^{2/5})}M^{-(k+1)/5}$
statistical queries.
\end{theorem}
\begin{proof}
We need to show that for any subset $S\subseteq[M]$ with $|S|=m$, 
$\p_S^A$ is far from $U_M$ in total variation distance. 
In particular, by Lemma~\ref{lem:TV-binary-lower-bound}, we have that
\begin{align*}
\dtv(U_M,\p_S^A)\ge\dtv\Big(U_M,U_M^{S,\frac{\delta}{\sqrt{m}}}\Big)-\dtv\Big(\p_S^A,U_M^{S,\frac{\delta}{\sqrt{m}}}\Big)\ge\Omega(\delta)-O(\epsilon)=\Omega(\delta).
\end{align*}
In addition, by our choice of $m$, we have that $\sqrt{2\tau}\le O(\delta)$. 
Therefore, we have that $\dtv\left(U_M,\p_S^A\right)\ge2\sqrt{2\tau}+\Omega(\delta)$. 
Applying Claim~\ref{clm:learning-to-testing} and Theorem~\ref{thm:sq-bin-testing} 
yields Theorem~\ref{thm:sq-bin-tv}.
\end{proof}

\begin{theorem}[SQ Lower Bound for Robust Mean Estimation of a Binary Product Distribution]\label{thm:sq-bin-mean}
Fix $0<c<1$ and $k$ to be a sufficiently large integer. Let $m,M\in\Z_+$ with $M=3m^{5/4}$.
Let $0<\eps<1/2$ and $\delta$ be a sufficiently small multiple of $\epsilon\sqrt{\log(1/\epsilon)}/k^2$. Let $\tau=\Theta(M^{-(k+1)/5}\delta)$.
Assume that $m>\max\Big(C'/\epsilon,\frac{k^2}{\log(1/\delta)}\Big)$ for some sufficiently large constant $C'>0$.
Then any $\mathrm{SQ}$ algorithm which is given access to a distribution $\p$ over $\{0,1\}^M$ which satisfies $\dtv\Big(\p,U_M^{S,\frac{\delta}{\sqrt{m}}}\Big)\le\epsilon$
for some unknown subset $S\subseteq[M]$ with $|S|=m$, outputs an estimate $\wh{\mu}$ with $\Big\|\wh{\mu}-\mu_M^{S,\frac{\delta}{\sqrt{m}}}\Big\|_2\le O(\delta)$ with probability at least $2/3$ must either make queries of accuracy better than $\sqrt{2\tau}$ or
must make at least $2^{\Omega(M^{2/5})}M^{-(k+1)/5}$
statistical queries.
\end{theorem}
\begin{proof}
Assume there is an algorithm that outputs an estimate $\wh{\mu}$ such that $\Big\|\wh{\mu}-\mu_M^{S,\frac{\delta}{\sqrt{m}}}\Big\|_2\le O(\delta)$ for some unknown subset $S\subseteq[M]$ with $|S|=m$.
Let $\q$ be the corresponding binary product distribution with mean vector $\wh{\mu}$.
Note that by our construction, both $\q$ and $U_M^{S,\frac{\delta}{\sqrt{m}}}$ are $c'$-balanced binary product distributions for some universal constant $c'>0$.
Therefore, by Lemma~\ref{lem:tv-balanced-bin}, we have that $\dtv(\q,\p_S^A)\le\dtv\Big(\q,U_M^{S,\frac{\delta}{\sqrt{m}}}\Big)+\dtv\Big(U_M^{S,\frac{\delta}{\sqrt{m}}},\p_S^A\Big)\le O\Big(\Big\|\wh{\mu}-\mu_M^{S,\frac{\delta}{\sqrt{m}}}\Big\|_2\Big)+O(\epsilon)\le O(\delta)$.
Applying Theorem~\ref{thm:sq-bin-tv} yields the result.
\end{proof}

\subsection{Proof of Proposition~\ref{prop:bin-moment-matching}}\label{ssec:bin-moment-matching}
In this section, we prove Proposition~\ref{prop:bin-moment-matching}.
We first introduce the following notation which will be used throughout this section.
For some fixed positive integer $n$ and $x\in[n]\cup\{0\}$, we consider the function $F_{n,x}(\delta)=\binom{n}{x}(1/2+\delta)^x(1/2-\delta)^{n-x},-1/2<\delta<1/2$. The first and second derivatives of $F_{n,x}(\delta)$ are given by the following fact:
\begin{fact}\label{fact:F-derivative}
For any positive integer $n$ and $x\in[n]\cup\{0\}$, we have that
\begin{align*}
F'_{n,x}(\delta)&=\binom{n}{x}(1/2+\delta)^{x-1}(1/2-\delta)^{n-x-1}(x-(1/2+\delta)n)=\frac{F_{n,x}(\delta)(x-(1/2+\delta)n)}{1/4-\delta^2},\\
F''_{n,x}(\delta)&=\binom{n}{x}(1/2+\delta)^{x-2}(1/2-\delta)^{n-x-2}(x^2-((2\delta+1)n-2\delta)x+(1/2+\delta)^2(n^2-n))\notag\\&=\frac{F_{n,x}(\delta)\left((x-(\delta+1/2)n)^2+2\delta(x-(\delta+1/2)n)+n(\delta^2-1/4)\right)}{(1/4-\delta^2)^2}.
\end{align*}
\end{fact}

We now pick $C=\Theta(\sqrt{(\log(1/\delta)/m)})$, where the hidden constant is sufficiently small.
Consider the interval $I_{C,m}=[(1/2-C)m,(1/2+C)m-1]$. Without loss of generality, we assume that the two endpoints of $I_{C,m}$ are integers.
We define the one-dimensional distribution $A$ to be:
\begin{itemize}
\item For $x\notin I_{C,m}$, we define $A(x)=\bin(m,1/2+\delta/\sqrt{m})(x)$.
\item For $x\in I_{C,m}$, we define $A(x)=\bin(m,1/2+\delta/\sqrt{m})(x)+\int_{x}^{x+1}p(t)dt$, where $p$ is a polynomial of degree at most $k$ satisfying
\begin{align}\label{eq:bin-moment-matching}
\sum_{x\in\Z\cap I_{C,m}}x^i\int_{x}^{x+1}p(t)dt=\sum_{x=0}^{m}(\bin(m,1/2)(x)-\bin(m,1/2+\delta/\sqrt{m})(x))x^i,
\end{align}
for $0\le i\le k$.
\end{itemize}
Applying Theorem~\ref{thm:moment-matching} with the family of functions $\{F_{m,x}(\delta)\}_{x\in[m]\cup\{0\}}$, we know that there is a unique polynomial $p$ of degree at most $k$ satisfying the above properties.
Then we need to show that with sufficiently large $m$ (depending on $\delta$), both the $L_1$ and $L_\infty$ norms of $p$ on $[(1/2-C)m,(1/2+C)m]$ are sufficiently small in order to make $A(x)$ non-negative and close to $\bin(m,1/2+\delta/\sqrt{m})$.
The main technical result of this section is the following lemma, which provides upper bounds on the $L_1$ and $L_\infty$ norms of $p$ on the interval $[(1/2-C)m,(1/2+C)m]$.

\begin{lemma}\label{lem:bin-norm-bound}
Let $k,m\in\Z_+$. Suppose $1\le k^2\le C_0C^2m$ for some universal constant $C_0>0$ sufficiently small and $m\ge C_1(\log(1/\delta))^3$ for some universal constant $C_1>0$ sufficiently large.
Then $\int_{(1/2-C)m}^{(1/2+C)m}|p(t)|dt\le O\left(\frac{\delta k^2}{C\sqrt{m}}\right)$ and $|p(t^*)|\le O\left(\frac{\delta k^{5/2}}{C^2m^{3/2}}\right)$, where $t^*=\arg\max_{t:|t-m/2|\le Cm}|p(t)|$.
\end{lemma}

Before we prove Lemma~\ref{lem:bin-norm-bound}, we first use it to prove our main Proposition~\ref{prop:bin-moment-matching}.
The following claim gives the upper bound of the ratio between the mass of $\bin(m,1/2)$ and $\bin(m,1/2+\delta/\sqrt{m})$.
\begin{claim}\label{clm:bin-ratio}
Let $m\in\Z_+$ and $x\in[m]\cup\{0\}$. For any $\delta>0$, we have that
\begin{align*}
\frac{F_{m,x}(\delta/\sqrt{m})}{F_{m,x}(0)}\le\exp\left(\frac{4\delta x}{\sqrt{m}}-2\delta\sqrt{m}\right).
\end{align*}
\end{claim}
\begin{proof}
By the fact $1+x\le e^x,\forall x\in\R$, we have that
\begin{align*}
\frac{F_{m,x}(\delta/\sqrt{m})}{F_{m,x}(0)}=\left(1+\frac{2\delta}{\sqrt{m}}\right)^x\left(1-\frac{2\delta}{\sqrt{m}}\right)^{m-x}\le\exp\left(\frac{4\delta x}{\sqrt{m}}-2\delta\sqrt{m}\right).
\end{align*}
\end{proof}

We now bound from above the desired $\chi^2$-divergence:
\begin{lemma}\label{lem:bin-chi}
We have that
\begin{align*}\chi^2(A,\mathrm{Bin}(m,1/2))\le O\left(\delta^2+\frac{\delta k^2\exp(4\delta C\sqrt{m})}{C\sqrt{m}}+\left(\frac{\delta k^2}{C\sqrt{m}}\right)\cdot\max_{x\in\Z\cap I_{C,m}}\frac{\int_{x}^{x+1}|p(t)|dt}{\bin(m,1/2)(x)}\right).
\end{align*}
\end{lemma}

\begin{proof}
Recalling $F_{m,x}(\delta)=\binom{m}{x}(1/2+\delta)^x(1/2-\delta)^{m-x}=\bin(m,1/2+\delta)(x),-1/2<\delta<1/2$, we have the following:
\begin{align*}
&\quad1+\chi^2(A,\mathrm{Bin}(m,1/2))=\sum_{x=0}^{m}\frac{A(x)^2}{F_{m,x}(0)}=\sum_{x=0}^{m}\frac{\left(F_{m,x}(\delta/\sqrt{m})+\mathbb{I}[x\in I_{C,m}]\int_{x}^{x+1}p(t)dt\right)^2}{F_{m,x}(0)}\\&=\sum_{x=0}^{m}\frac{F^2_{m,x}(\delta/\sqrt{m})}{F_{m,x}(0)}+2\sum_{x\in\Z\cap I_{C,m}}\frac{F_{m,x}(\delta/\sqrt{m})\int_{x}^{x+1}p(t)dt}{F_{m,x}(0)}+\sum_{x\in\Z\cap I_{C,m}}\frac{\left(\int_{x}^{x+1}p(t)dt\right)^2}{F_{m,x}(0)}\;.
\end{align*}
For the first term, we have that
\begin{align*}
\sum_{x=0}^{m}\frac{F^2_{m,x}(\delta/\sqrt{m})}{F_{m,x}(0)}&=2^{-m}\sum_{x=0}^{m}\binom{m}{x}\left(1+\frac{2\delta}{\sqrt{m}}\right)^{2x}\left(1-\frac{2\delta}{\sqrt{m}}\right)^{2(m-x)}\\&=2^{-m}\left(\left(1+\frac{2\delta}{\sqrt{m}}\right)^2+\left(1-\frac{2\delta}{\sqrt{m}}\right)^2\right)^m=\left(1+\frac{4\delta^2}{m}\right)^m\le1+O(\delta^2)\;.
\end{align*}
For the second term, by Claim~\ref{clm:bin-ratio}, we have that
\begin{align*}
\sum_{x\in\Z\cap I_{C,m}}\frac{F_{m,x}(\delta/\sqrt{m})\int_{x}^{x+1}p(t)dt}{F_{m,x}(0)}&\le\sum_{x\in\Z\cap I_{C,m}}\exp\left(\frac{4\delta x}{\sqrt{m}}-2\delta\sqrt{m}\right)\left|\int_{x}^{x+1}p(t)dt\right|\\&\le\exp\left(4\delta C\sqrt{m}\right)\sum_{x\in\Z\cap I_{C,m}}\left|\int_{x}^{x+1}p(t)dt\right|\\&\le\exp(4\delta C\sqrt{m})\int_{(1/2-C)m}^{(1/2+C)m}|p(t)|dt\\&\le O\left(\frac{\delta k^2\exp(4\delta C\sqrt{m})}{C\sqrt{m}}\right),
\end{align*}
where the last inequality follows from Lemma~\ref{lem:bin-norm-bound}.
Finally for the third term, we have that
\begin{align*}
\sum_{x\in\Z\cap I_{C,m}}\frac{\left(\int_{x}^{x+1}p(t)dt\right)^2}{F_{m,x}(0)}&\le\int_{(1/2-C)m}^{(1/2+C)m}|p(x)|dx\cdot\max_{x\in\Z\cap I_{C,m}}\frac{\int_{x}^{x+1}|p(t)|dt}{F_{m,x}(0)}\\&\le O\left(\frac{\delta k^2}{C\sqrt{m}}\right)\cdot\max_{x\in\Z\cap I_{C,m}}\frac{\int_{x}^{x+1}|p(t)|dt}{F_{m,x}(0)},
\end{align*}
where the last inequality follows from Lemma~\ref{lem:bin-norm-bound}.
Combining the above results together completes the proof.
\end{proof}

We are now ready to prove Proposition~\ref{prop:bin-moment-matching}.
We need to pick $C$ appropriately and check the bounds on $k$ needed for $A(x)$ to satisfy the necessary properties.
\begin{proof}[Proof of Proposition~\ref{prop:bin-moment-matching}]
Let $C=\Theta(\sqrt{\log(1/\delta)/m})$ with the hidden constant sufficiently small.
If $k^2\ge C\sqrt{m}$, we pick $A=\mathrm{Bin}(m,1/2)$ and obtain $\dtv(A,\mathrm{Bin}(m,1/2+\delta/\sqrt{m}))\le O(\delta)\le O\left(\frac{k^2\delta}{\sqrt{\log(1/\delta)}}\right)$.
Thus, we assume that $k^2\le C\sqrt{m}$.
In this way, to apply Lemma~\ref{lem:bin-norm-bound}, we need $k^2\le C_0C^2m$ for some universal constant $C_0$ sufficiently small, which will be satisfied as long as $\delta\le\exp(-1/C_0^2)$.

We first show that $A(x)$ is indeed a distribution over $[m]\cup\{0\}$.
By definition, $A(x)$ is nonnegative outside the interval $I_{C,m}$.
For $x\in\Z\cap I_{C,m}$, we apply Lemma~\ref{lem:bin-norm-bound} to obtain
\begin{align*}
A(x)&=F_{m,x}(\delta/\sqrt{m})+\int_{x}^{x+1}p(t)dt\ge F_{m,x}(\delta/\sqrt{m})-|p(t^*)|\\&=2^{-m}\binom{m}{x}\left(1+\frac{2\delta}{\sqrt{m}}\right)^x\left(1-\frac{2\delta}{\sqrt{m}}\right)^{m-x}-|p(t^*)|\\&\ge2^{-m}\binom{m}{(1/2-C)m}\left(1+\frac{2\delta}{\sqrt{m}}\right)^{\left(\frac{1}{2}-C\right)m}\left(1-\frac{2\delta}{\sqrt{m}}\right)^{\left(\frac{1}{2}+C\right)m}-O\left(\frac{\delta k^{5/2}}{C^2m^{3/2}}\right)\\&=2^{-m}\binom{m}{(1/2-C)m}\left(1-\frac{4\delta^2}{m}\right)^{\left(\frac{1}{2}-C\right)m}\left(1-\frac{2\delta}{\sqrt{m}}\right)^{2Cm}-O\left(\frac{\delta k^{5/2}}{C^2m^{3/2}}\right).
\end{align*}
Let $H(x)=-x\log x-(1-x)\log(1-x)$ denote the binary entropy function.
Now applying Fact~\ref{fact:bin-bound-2} and the fact $e^{-2x}\le1-x,\forall x\in\left[0,\frac{\ln2}{2}\right]$ yields
\begin{align*}
A(x)&\ge\frac{2^{m(H(1/2+C)-1)}}{\sqrt{(2-8C^2)m}}\cdot\exp(-8\delta^2(1/2-C))\cdot\exp(-8\delta C\sqrt{m})-O\left(\frac{\delta k^{5/2}}{C^2m^{3/2}}\right)\\&
\ge\frac{2^{m(H(1/2+C)-1)}}{\sqrt{(2-8C^2)m}}\cdot\exp\left(-4\delta^2-8\delta C\sqrt{m}\right)-O\left(\frac{\delta k^{5/2}}{C^2m^{3/2}}\right)\\&\ge\frac{\exp\left(-O(C^2m)-8\delta C\sqrt{m}\right)}{\sqrt{m}}-O\left(\frac{\delta k^{5/2}}{C^2m^{3/2}}\right)\\&\ge\frac{\exp(-O((C\sqrt{m}+\delta)^2))}{\sqrt{m}}-O\left(\frac{\delta k^{5/2}}{C^2m^{3/2}}\right),
\end{align*}
where the third inequality follows from the Taylor expansion of $H(1/2+C)-H(1/2)$ up to second order terms.
Note that $k^2\le C\sqrt{m}$, where $C=\Theta(\sqrt{\log(1/\delta)/m})$ for some sufficiently small hidden constant in $\Theta$, we have that $\frac{\delta k^{5/2}}{C^2m}\le O(\delta(\log(1/\delta))^{-3/8})$ and $\exp(-O((C\sqrt{m}+\delta)^2))\ge\exp\big(-\big(\sqrt{\log(1/\delta)}/2+\delta\big)^2\big)\ge\delta$.
Therefore, we have that
$$A(x)\ge\frac{\exp(-O((C\sqrt{m}+\delta)^2))}{\sqrt{m}}-O\left(\frac{\delta k^{5/2}}{C^2m^{3/2}}\right)\ge0,\forall x\in\Z\cap I_{C,m}.$$

In addition, by equation~\eqref{eq:bin-moment-matching}, we know that
\begin{align*}
\sum_{x=0}^{m}A(x)&=\sum_{x=0}^{m}\left(F_{m,x}(\delta/\sqrt{m})+\mathbb{I}[x\in I_{C,m}]\int_{x}^{x+1}p(t)dt\right)\\&=\sum_{x=0}^{m}F_{m,x}(\delta/\sqrt{m})+\int_{(1/2-C)m}^{(1/2+C)m}p(t)dt=1,
\end{align*}
which implies that the distribution $A$ is well-defined.
Furthermore, by Equation~\eqref{eq:bin-moment-matching}, we can show that $A$ matches the first $k$ moments of $\bin(m,1/2)$ as follows:
\begin{align*}
\E_{X\sim A}[X^i]&=\sum_{x=0}^{m}{A(x)x^i}=\sum_{x=0}^{m}\left(F_{m,x}(\delta/\sqrt{m})+\mathbb{I}[x\in I_{C,m}]\int_{x}^{x+1}p(t)dt\right)x^i\\&=\sum_{x=0}^{m}F_{m,x}(\delta/\sqrt{m})x^i+\sum_{x\in\Z\cap I_{C,m}}x^i\int_{x}^{x+1}p(t)dt\\&=\sum_{x=0}^{m}F_{m,x}(0)x^i=\E_{X\sim\bin(m,1/2)}[X^i].
\end{align*}

From previous calculation, we have that $A(x)\ge F_{m,x}(\delta/\sqrt{m})-|p(t^*)|\ge0,\forall x\in\Z\cap I_{C,m}$, which implies that for every $x\in\Z\cap I_{C,m}$,
$$|p(t^*)|\le F_{m,x}(\delta/\sqrt{m})\le\exp\left(\frac{4\delta x}{\sqrt{m}}-2\delta\sqrt{m}\right)F_{m,x}(0)\le\exp\left(4\delta C\sqrt{m}\right)F_{m,x}(0),$$
where the second inequality follows from Claim~\ref{clm:bin-ratio}.
Therefore, by Lemma~\ref{lem:bin-chi}, we have that
\begin{align*}
\chi^2(A,\mathrm{Bin}(m,1/2))&\le O\left(\delta^2+\frac{\delta k^2\exp(4\delta C\sqrt{m})}{C\sqrt{m}}+\left(\frac{\delta k^2}{C\sqrt{m}}\right)\cdot\max_{x\in\Z\cap I_{C,m}}\frac{\int_{x}^{x+1}|p(t)|dt}{\bin(m,1/2)(x)}\right)\\&\le O\left(\delta^2+\delta\left(\exp(4\delta C\sqrt{m})+\frac{|p(t^*)|}{F_{m,x}(0)}\right)\right)\le O\left(\delta^2+2\delta\exp\left(4\delta C\sqrt{m}\right)\right)\\&\le O\left(\delta^2+\delta\left(1+O(\delta\sqrt{\log(1/\delta)})\right)\right)=O(\delta),
\end{align*}
where we apply the fact $e^x\le1+2x,\forall x\in\left[0,\ln2\right]$.

To bound the total variation distance $\dtv(A,\mathrm{Bin}(1/2+\delta/\sqrt{m}))$, we apply Lemma~\ref{lem:bin-norm-bound} to obtain
\begin{align*}
\dtv(A,\mathrm{Bin}(1/2+\delta/\sqrt{m}))&=\sum_{x\in I_{C,m}}\left|\int_{x}^{x+1}p(t)dt\right|\le\int_{(1/2-C)m}^{(1/2+C)m}|p(t)|dt\\&\le O\left(\frac{\delta k^2}{C\sqrt{m}}\right)=O\left(\frac{\delta k^2}{\sqrt{\log(1/\delta)}}\right).
\end{align*}
This completes the proof of Proposition~\ref{prop:bin-moment-matching}.
\end{proof}

\paragraph{Proof of Lemma~\ref{lem:bin-norm-bound}}
By Theorem~\ref{thm:moment-matching}, we have that
\begin{align*}
|a_i|\le\left(\frac{2i+1}{2Cm}\right)\left(\beta_i+O\left(\frac{i^2}{Cm}\right)\int_{(1/2-C)m}^{(1/2+C)m}|p(t)|dt\right),
\end{align*}
for all $1\le i\le k$, where $\beta_i=\left|\sum_{x=0}^{m}(F_{m,x}(0)-F_{m,x}(\delta/\sqrt{m}))P_i\left(\frac{x-m/2}{Cm}\right)\right|$.
To get an upper bound for the $L_1$ and $L_\infty$ norms of the polynomial $p$ over $I_{C,m}$, we only need to upper bound the quantity $\beta_i$.
\begin{lemma}\label{lem:bin-beta}
If $k^2\le C_0C^2m$ for some universal constant $C_0>0$ sufficiently small, then $\beta_i\le O\left(\frac{\delta}{C}\sqrt{\frac{i}{m}}\right),\forall1\le i\le k$.
\end{lemma}
We assume $k^2\le C_0C^2m$ for some universal constant $C_0>0$ sufficiently small.
First, we apply Taylor's theorem to expand $F_{m,x}(\delta/\sqrt{m})-F_{m,x}(0)$ up to second order terms:
\begin{align}\label{eq:F-taylor}
\beta_i&=\left|\sum_{x=0}^{m}\left(F_{m,x}(0)-F_{m,x}(\delta/\sqrt{m})\right)P_i\left(\frac{x-m/2}{Cm}\right)\right|\notag\\&=\left|\sum_{x=0}^{m}\left(\frac{\delta F'_{m,x}(0)}{\sqrt{m}}+\frac{F''_{m,x}(\delta_x)\delta^2}{2m}\right)P_i\left(\frac{x-m/2}{Cm}\right)\right|\notag\\&\le\underbrace{\left|\sum_{x=0}^{m}\left(\frac{\delta F'_{m,x}(0)}{\sqrt{m}}\right)P_i\left(\frac{x-m/2}{Cm}\right)\right|}_{\beta'_i}+\underbrace{\left|\sum_{x=0}^{m}\left(\frac{F''_{m,x}(\delta_x)\delta^2}{2m}\right)P_i\left(\frac{x-m/2}{Cm}\right)\right|}_{\beta''_i}\;,
\end{align}
where for any $x\in[m]\cup\{0\}$, $\delta_x=\wt{\delta_x}/\sqrt{m}$ for some $\wt{\delta_x}\in[0,\delta]$.

Hence, in order to bound $\beta_i$, it suffices to bound the terms $\beta'_i:=\frac{\delta}{\sqrt{m}}\left|\sum_{x=0}^{m}F'_{m,x}(0)P_i\left(\frac{x-m/2}{Cm}\right)\right|$ and $\beta_i'':=\frac{\delta^2}{2m}\left|\sum_{x=0}^{m}F''_{m,x}(\delta_x)P_i\left(\frac{x-m/2}{Cm}\right)\right|$.
This is done in the following lemmas.

\begin{lemma}\label{lem:bin-first-order}
We have that $\beta'_i\le O\left(\max\left(\frac{\delta}{C}\sqrt{\frac{i}{m}},\frac{\delta i^{3/2}}{C^2m}\right)\right),\forall 1\le i\le k$.
\end{lemma}
\begin{proof}
If $i$ is odd, we can rewrite Fact~\ref{fact:legendre-poly} (v) in ascending order of terms, by using change of variables to obtain
\begin{align*}
\left|P_i\left(\frac{x-m/2}{Cm}\right)\right|&\le2^{-i}\sum_{j=1}^{(i+1)/2}\binom{i}{(i-1)/2+j}\binom{i+2j-1}{2j-1}\left|\frac{x-m/2}{Cm}\right|^{2j-1}\\&\le O\left(\frac{1}{\sqrt{i}}\right)\sum_{j=1}^{(i+1)/2}\frac{(i+2j-1)^{2j-1}}{(2j-1)!}\left|\frac{x-m/2}{Cm}\right|^{2j-1},
\end{align*}
where the second inequality follows from Fact~\ref{fact:bin-bound-1}. Note that $\|X-m/2\|_{\psi_2}\le O(\sqrt{m})$ for $X\sim\bin(m,1/2)$, applying Fact~\ref{fact:Gaussian-decay} and Fact~\ref{fact:F-derivative} yields
\begin{align*}
\beta'_i&=\frac{\delta}{\sqrt{m}}\left|\sum_{x=0}^{m}F'_{m,x}(0)P_i\left(\frac{x-m/2}{Cm}\right)\right|=\frac{\delta2^{2-m}}{\sqrt{m}}\sum_{x=0}^{m}\left|\binom{m}{x}(x-m/2)P_i\left(\frac{x-m/2}{Cm}\right)\right|\\&\le\frac{\delta2^{2-m}}{\sqrt{m}}\sum_{j=1}^{(i+1)/2}O\left(\frac{1}{\sqrt{i}}\right)\left(\frac{(i+2j-1)^{2j-1}}{(2j-1)!(Cm)^{2j-1}}\right)\sum_{x=0}^{m}\binom{m}{x}(x-m/2)^{2j}\\&=\frac{4\delta}{\sqrt{m}}\sum_{j=1}^{(i+1)/2}O\left(\frac{1}{\sqrt{i}}\right)\left(\frac{(i+2j-1)^{2j-1}\E_{X\sim\bin(m,1/2)}[(X-m/2)^{2j}]}{(2j-1)!(Cm)^{2j-1}}\right)\\&\le O\left(\frac{\delta}{\sqrt{mi}}\right)\sum_{j=1}^{(i+1)/2}\frac{(i+2j-1)^{2j-1}(2j!)(O(m))^j}{(2j-1)!(Cm)^{2j-1}}\\&\le O\left(\frac{\delta}{\sqrt{i}}\right)\sum_{j=1}^{\infty}\left(O\left(\frac{i}{C\sqrt{m}}\right)\right)^{2j-1}\le O\left(\frac{\delta}{C}\sqrt{\frac{i}{m}}\right).
\end{align*}

If $i$ is even, applying Fact~\ref{fact:F-derivative} and Fact~\ref{fact:legendre-poly} (v) by using change of variables yields
\begin{align*}
\beta'_i&=\frac{\delta}{\sqrt{m}}\left|\sum_{x=0}^{m}F'_{m,x}(0)P_i\left(\frac{x-m/2}{Cm}\right)\right|=\frac{\delta2^{2-m}}{\sqrt{m}}\left|\sum_{x=0}^{m}\binom{m}{x}(x-m/2)P_i\left(\frac{x-m/2}{Cm}\right)\right|\\&=\frac{\delta2^{2-m}}{\sqrt{m}}\left|\sum_{x=0}^{m}\binom{m}{x}(x-m/2)2^{-i}\sum_{j=0}^{i/2}(-1)^j\binom{i}{j}\binom{2i-2j}{i}\left(\frac{x-m/2}{Cm}\right)^{i-2j}\right|\\&=\frac{\delta2^{2-m}}{\sqrt{m}}\left|\sum_{x=0}^{m}\binom{m}{x}(x-m/2)2^{-i}\sum_{j=0}^{i/2-1}(-1)^j\binom{i}{j}\binom{2i-2j}{i}\left(\frac{x-m/2}{Cm}\right)^{i-2j}\right|\\&\le\frac{\delta2^{2-m}}{\sqrt{m}}\sum_{x=0}^{m}\binom{m}{x}|x-m/2|2^{-i}\sum_{j=1}^{i/2}\binom{i}{i/2+j}\binom{i+2j}{2j}\left|\frac{x-m/2}{Cm}\right|^{2j}\\&\le\frac{\delta2^{2-m}}{\sqrt{m}}\sum_{j=1}^{i/2}\sum_{x=0}^{m}\binom{m}{x}|x-m/2|^{2j+1}O\left(\frac{1}{\sqrt{i}}\right)\left(\frac{(i+2j)^{2j}}{(2j)!(Cm)^{2j}}\right)\\&=\frac{4\delta}{\sqrt{m}}\sum_{j=1}^{i/2}O\left(\frac{1}{\sqrt{i}}\right)\left(\frac{(i+2j)^{2j}\E_{X\sim\bin(m,1/2)}[|X-m/2|^{2j+1}]}{(2j)!(Cm)^{2j}}\right)\\&\le O\left(\frac{\delta}{\sqrt{mi}}\right)\sum_{j=1}^{i/2}\frac{(i+2j)^{2j}(2j+1)(O(m))^{j+1/2}j!}{(2j)!(Cm)^{2j}}\\&\le O\left(\frac{\delta}{\sqrt{i}}\right)\sum_{j=1}^{\infty}\left(O\left(\frac{i}{C\sqrt{m}}\right)\right)^{2j}\le O\left(\frac{\delta i^{3/2}}{C^2m}\right),
\end{align*}
where the second inequality follows from Fact~\ref{fact:bin-bound-1} and the third inequality follows from Fact~\ref{fact:Gaussian-decay} and the fact that $\|X-m/2\|_{\psi_2}\le O(\sqrt{m})$ for $X\sim\bin(m,1/2)$.
\end{proof}
For the quantity $\beta''_i$, we have the following lemma.
\begin{lemma}\label{lem:bin-second-order}
We have that $\beta''_i=O(\delta^2),\forall 1\le i\le k$.
\end{lemma}
\begin{proof}
By Fact~\ref{fact:F-derivative}, we have that
\begin{align*}
\beta''_i&=\left|\sum_{x=0}^{m}F''_{m,x}(\delta_x)P_i\left(\frac{x-m/2}{Cm}\right)\right|\\&=\frac{1}{(1/4-\delta_x^2)^2}\left|\sum_{x=0}^{m}F_{m,x}(\delta_x)\left((x-(\delta_x+1/2)m)^2+2\delta_x(x-(\delta_x+1/2)m)+m(\delta_x^2-1/4)\right)P_i\left(\frac{x-m/2}{Cm}\right)\right|\\&=O\left(\left|\sum_{x=0}^{m}F_{m,x}(\delta_x)\left((x-(\delta_x+1/2)m)^2+2\delta_x(x-(\delta_x+1/2)m)+m(\delta_x^2-1/4)\right)P_i\left(\frac{x-m/2}{Cm}\right)\right|\right).
\end{align*}
We separate the above sum into $x\in\Z\cap I_{C,m}$ and $x\in[m]\cup\{0\}\setminus I_{C,m}$.
We are able to use Fact~\ref{fact:legendre-poly} (iii) to bound the sum for $x\in\Z\cap I_{C,m}$, as follows:
\begin{align*}
&\quad\left|\sum_{x\in\Z\cap I_{C,m}}F_{m,x}(\delta_x)\left((x-(\delta_x+1/2)m)^2+2\delta_x(x-(\delta_x+1/2)m)+m(\delta_x^2-1/4)\right)P_i\left(\frac{x-m/2}{Cm}\right)\right|\\&\le\sum_{x=0}^{m}F_{m,x}(\delta_x)\left|(x-(\delta_x+1/2)m)^2-\var_{X\sim\bin(m,1/2+\delta_x)}[X]+2{\delta_x}(x-(\delta_x+1/2)m)\right|\\&\le2\var_{X\sim\bin(m,1/2+\delta_x)}[X]+2{\delta_x}\E_{X\sim\bin(m,1/2+\delta_x)}\left[\left|X-\E_{X\sim\bin(m,1/2+\delta_x)}[X]\right|\right]\\&\le O\left(\var_{X\sim\bin(m,1/2+\delta_x)}[X]+\delta_x\sqrt{\var_{X\sim\bin(m,1/2+\delta_x)}[X]}\right)\\&\le O\left(m+\delta\right)\le O(m).
\end{align*}

Now we bound the sum over $x\in\overline{I_{C,m}}$, where $\overline{I_{C,m}}=[m]\cup\{0\}\setminus I_{C,m}$.
Note that for $X\sim\bin(m,1/2+\delta_x)$, by Fact~\ref{fact:Gaussian-exp-1}, we have that
\begin{align*}
\left\|\frac{4(X-m/2)}{Cm}\right\|_{\psi_2}&=O\left(\left\|\frac{X-(1/2+\delta_x)m}{Cm}\right\|_{\psi_2}+\left\|\delta_x/C\right\|_{\psi_2}\right)=O\left(\left\|\frac{X-(1/2+\delta_x)m}{Cm}\right\|_{\psi_2}+\delta_x/C\right)\\&\le O\left(\frac{1}{C\sqrt{m}}+\frac{\delta}{C\sqrt{m}}\right)=O\left(\frac{1}{C\sqrt{m}}\right).
\end{align*}
Therefore, applying Fact~\ref{fact:legendre-poly} (vi) yields
\begin{align*}
&\quad\left|\sum_{x\in\overline{I_{C,m}}}F_{m,x}(\delta_x)(x-(\delta_x+1/2)m)^2+2\delta_x(x-(\delta_x+1/2)m)+m(\delta_x^2-1/4)P_i\left(\frac{x-m/2}{Cm}\right)\right|\\&\le\sum_{x=0}^{m}F_{m,x}(\delta_x)\cdot\left|(x-(\delta_x+1/2)m)^2+2\delta_x(x-(\delta_x+1/2)m)+m(\delta_x^2-1/4)\right|\cdot\left|\frac{4(x-m/2)}{Cm}\right|^i\\&\le O(m)\cdot\mathop{\E}_{X\sim\bin(m,1/2+\delta_x)}\left[\left|\frac{4(X-m/2)}{Cm}\right|^i\right]+\mathop{\E}_{X\sim\bin(m,1/2+\delta_x)}\left[(X-(\delta_x+1/2)m)^2\cdot\left|\frac{4(X-m/2)}{Cm}\right|^i\right]\\&\quad+2\delta_x\mathop{\E}_{X\sim\bin(m,1/2+\delta_x)}\left[|X-(\delta_x+1/2)m|\cdot\left|\frac{4(X-m/2)}{Cm}\right|^i\right]\\&\le O(m)\mathop{\E}_{X\sim\bin(m,1/2+\delta_x)}\left[\left|\frac{4(X-m/2)}{Cm}\right|^i\right]\\&\quad+\sqrt{\mathop{\E}_{X\sim\bin(m,1/2+\delta_x)}\left[(X-(\delta_x+1/2)m)^4\right]\cdot\mathop{\E}_{X\sim\bin(m,1/2+\delta_x)}\left[\left|\frac{4(X-m/2)}{Cm}\right|^{2i}\right]}\\&\quad+2\delta_x\sqrt{\mathop{\E}_{X\sim\bin(m,1/2+\delta_x)}\left[(X-(\delta_x+1/2)m)^2\right]\cdot\mathop{\E}_{X\sim\bin(m,1/2+\delta_x)}\left[\left|\frac{4(X-m/2)}{Cm}\right|^{2i}\right]}\\&\le O(m)\cdot\left(O\left(\frac{1}{C}\sqrt{\frac{i}{m}}\right)\right)^i+\sqrt{O(m^2)\cdot\left(O\left(\frac{i}{C^2m}\right)\right)^i}+2\delta_x\sqrt{O(m)\cdot\left(O\left(\frac{i}{C^2m}\right)\right)^i}\\&\le O(m),
\end{align*}
where the third inequality follows from Cauchy-Schwarz and the fourth inequality follows from Fact~\ref{fact:Gaussian-decay}.
Combine the above results together, we have that 
\begin{align*}
\beta_i''=\frac{\delta^2}{2m}\left|\sum_{x=0}^{m}F''_{m,x}(\delta_x)P_i\left(\frac{x-m/2}{Cm}\right)\right|\le\left(\frac{\delta^2}{2m}\right)\cdot O(m)=O(\delta^2).
\end{align*}
\end{proof}

\begin{proof}[Proof of Lemma~\ref{lem:bin-beta}]
Since $k^2\le C^2m$, by Lemma~\ref{lem:bin-first-order}, Lemma~\ref{lem:bin-second-order} and equation~\eqref{eq:F-taylor}, we have that $\beta_i\le\beta'_i+\beta''_i\le O\left(\frac{\delta}{C}\sqrt{\frac{i}{m}}+\delta^2\right)=O\left(\frac{\delta}{C}\sqrt{\frac{i}{m}}\right)$.
\end{proof}

We are now ready to prove Lemma~\ref{lem:bin-norm-bound}.
\begin{proof}[Proof of Lemma~\ref{lem:bin-norm-bound}]
By Theorem~\ref{thm:moment-matching}, we have that
\begin{align*}
|p(t^*)|&\le\sum_{i=1}^{k}|a_i|\le\sum_{i=1}^{k}\left(\frac{2i+1}{2Cm}\right)\beta_i+\sum_{i=1}^{k}\left(\frac{2i+1}{2Cm}\right)O\left(\frac{i^2}{Cm}\right)\int_{(1/2-C)m}^{(1/2+C)m}|p(t)|dt\\&\le\sum_{i=1}^{k}\left(\frac{2i+1}{2Cm}\right)\beta_i+|p(t^*)|\sum_{i=1}^{k}\left(\frac{2i+1}{2}\right)O\left(\frac{i^2}{Cm}\right)\\&\le\sum_{i=1}^{k}\left(\frac{2i+1}{2Cm}\right)\beta_i+O\left(\frac{k^4}{Cm}\right)|p(t^*)|,
\end{align*}
where the first inequality follows from Fact~\ref{fact:legendre-poly} (iii). Similarly, by Theorem~\ref{thm:moment-matching}, we have that
\begin{align*}
\int_{(1/2-C)m}^{(1/2+C)m}|p(t)|dt&\le\sum_{i=1}^{k}|a_i|\int_{(1/2-C)m}^{(1/2+C)m}\left|P_i\left(\frac{t-m/2}{Cm}\right)\right|dt\\&\le Cm\sum_{i=1}^{k}\left(\frac{2i+1}{2Cm}\right)\left(\beta_i+O\left(\frac{i^2}{Cm}\right)\int_{(1/2-C)m}^{(1/2+C)m}|p(t)|dt\right)\int_{-1}^{1}|P_i(y)|dy\\&\le\sum_{i=1}^{k}O(\sqrt{i})\beta_i+\sum_{i=1}^{k}O(\sqrt{i})O\left(\frac{i^2}{Cm}\right)\int_{(1/2-C)m}^{(1/2+C)m}|p(t)|dt\\&\le\sum_{i=1}^{k}O(\sqrt{i})\beta_i+O\left(\frac{k^{7/2}}{Cm}\right)\int_{(1/2-C)m}^{(1/2+C)m}|p(t)|dt,
\end{align*}
where the third inequality follows from Fact~\ref{fact:legendre-poly} (vii). By our assumption on $k,C,m,\delta$, we know that $\frac{k^{7/2}}{Cm}\le\frac{k^4}{Cm}\le1/2$. Therefore, by Lemma~\ref{lem:bin-beta}, we have that
\begin{align*}
&|p(t^*)|\le2\sum_{i=1}^{k}\left(\frac{2i+1}{2Cm}\right)\beta_i\le\sum_{i=1}^{k}\left(\frac{2i+1}{2Cm}\right)O\left(\frac{\delta}{C}\sqrt{\frac{i}{m}}\right)\le O\left(\frac{\delta k^{5/2}}{C^2m^{3/2}}\right),\\&
\int_{(1/2-C)m}^{(1/2+C)m}|p(t)|dt\le2\sum_{i=1}^{k}O(\sqrt{i})\beta_i\le2\sum_{i=1}^{k}O(\sqrt{i})O\left(\frac{\delta}{C}\sqrt{\frac{i}{m}}\right)\le O\left(\frac{\delta k^2}{C\sqrt{m}}\right).
\end{align*}
This completes the proof.
\end{proof}

\subsection{Proof of Lemma~\ref{lem:TV-binary-lower-bound}}\label{ssec:TV-binary-lower-bound}
Let $n=|S|\le m$. Define $f(\bx)=\sum_{i\in S}x_i,\forall \bx\in\{0,1\}^M$. For $\bX\sim U_M$ and $\bY\sim U_M^{S,\frac{\delta}{\sqrt{m}}}$, by the data processing inequality, we have that
\begin{align*}
\dtv\Big(U_M,U_M^{S,\frac{\delta}{\sqrt{m}}}\Big)\ge\dtv(f(\bX),f(\bY))=\dtv(\mathrm{Bin}(n,1/2),\mathrm{Bin}(n,1/2+\delta/\sqrt{m})).
\end{align*}
Recalling that $F_{n,x}(\delta)=\binom{n}{x}(1/2+\delta)^x(1/2-\delta)^{n-x}$, we can write
\begin{align*}
\dtv(\mathrm{Bin}(n,1/2),\mathrm{Bin}(n,1/2+\delta/\sqrt{m}))&=\frac{1}{2}\sum_{x=0}^{n}|F_{n,x}(\delta/\sqrt{m})-F_{n,x}(0)|=\frac{1}{2}\sum_{x=0}^{n}\left|\frac{\delta F'_{n,x}(0)}{\sqrt{m}}+\frac{F''_{n,x}(\delta_x)\delta^2}{2m}\right|,
\end{align*}
where for any $x\in[m]\cup\{0\}$, $\delta_x=\wt{\delta_x}/\sqrt{m}$ for some $\wt{\delta_x}\in(0,\delta)$. Applying Fact~\ref{fact:F-derivative} and Fact~\ref{fact:E-lower} yields
\begin{align*}
\sum_{x=0}^{n}|F'_{n,x}(0)|&=4\sum_{x=0}^{n}F_{n,x}(0)|x-n/2|=4\E_{X\sim\bin(n,1/2)}[|X-n/2|]\\&\ge\frac{4\E_{X\sim\bin(n,1/2)}[(X-n/2)^2]^{3/2}}{\E_{X\sim\bin(n,1/2)}[(X-n/2)^4]^{1/2}}\\&=\frac{4(n/4)^{3/2}}{\sqrt{(n/4)(1+(3n-6)/4)}}\\&=\Theta(\sqrt{n}).
\end{align*}
In addition, by Fact~\ref{fact:F-derivative}, we have that
\begin{align*}
\sum_{x=0}^{n}|F''_{n,x}(\delta_x)|&=\frac{1}{(1/4-\delta_x^2)^2}\sum_{x=0}^{n}|F_{n,x}(\delta_x)\left((x-(\delta_x+1/2)n)^2+2\delta_x(x-(\delta_x+1/2)n)+n(\delta_x^2-1/4)\right)|\\&\le\frac{1}{(1/4-\delta_x^2)^2}\left(2\var_{X\sim\bin\left(n,\frac{1}{2}+\delta_x\right)}[X]+2\delta_x\E_{X\sim\bin\left(n,\frac{1}{2}+\delta_x\right)}\left[\left|X-\E_{X\sim\bin\left(n,\frac{1}{2}+\delta_x\right)}[X]\right|\right]\right)\\&\le\frac{1}{(1/4-\delta_x^2)^2}\left(2\var_{X\sim\bin\left(n,\frac{1}{2}+\delta_x\right)}[X]+2\delta_x\sqrt{\var_{X\sim\bin\left(n,\frac{1}{2}+\delta_x\right)}[X]}\right)\\&\le O\left(n+\delta\sqrt{\frac{n}{m}}\right)\le O(n).
\end{align*}
Therefore, we have that
\begin{align*}
&\quad\dtv(\mathrm{Bin}(n,1/2),\mathrm{Bin}(n,1/2+\delta/\sqrt{m}))=\frac{1}{2}\sum_{x=0}^{n}\left|\frac{\delta F'_{n,x}(0)}{\sqrt{m}}+\frac{F_{n,x}(\delta_x)\delta^2}{2m}\right|\\&\ge\frac{\delta}{2\sqrt{m}}\sum_{x=0}^{n}|F'_{n,x}(0)|-\frac{\delta^2}{4m}\sum_{x=0}^{n}|F''_{n,x}(\delta_x)|\ge\Omega\left(\delta\sqrt{\frac{n}{m}}\right)-O\left(\frac{\delta^2n}{m}\right)=\Omega(\delta)-O(\delta^2)=\Omega(\delta).
\end{align*}

\section{Omitted Statements and Proofs from Section~\ref{ssec:sq-ising}}\label{app:sq-ising}
We provide the hardness result for robust learning of an unknown ferromagnetic high temperature Ising model here.
In order to make the distributions in our family far from the reference distribution $U_M$ in total variation distance,
we need higher dimension $m,M$ compared with the hardness result for robust hypothesis testing.

\begin{theorem}[SQ Lower Bound for Robust Learning of an Unknown Ising Model]\label{thm:sq-ising}
Fix $0<c<1$ and $k$ to be a sufficiently large integer. Let $m,M\in\Z_+$ with $M=3m^{5/4}$. Let $0<\eps<1/2$
and $\delta$ be a sufficiently small multiple of $\epsilon\log(1/\epsilon)/k^3$. Let $\tau=\Theta(M^{-(k+1)/5}\delta)$.
Assume that $m>\max\Big(C'/\epsilon,\frac{k^2}{\log(1/\delta)}\Big)$ for some sufficiently large constant $C'>0$.
Then any $\mathrm{SQ}$ algorithm which is given access to a distribution $\p$ over $\{0,1\}^M$ which satisfies $\dtv\Big(\p,Q_M^{S,\frac{\delta}{m}}\Big)\le\epsilon$
for some unknown subset $S\subseteq[M]$ with $|S|=m$, outputs a hypothesis $\q$ with $\dtv(\q,\p)\le O(\delta)$ with probability at least $2/3$ must either make queries of accuracy better than $\sqrt{2\tau}$ or
must make at least $2^{\Omega(M^{2/5})}M^{-(k+1)/5}$
statistical queries.
\end{theorem}
\begin{proof}
We need to show that for any subset $S\subseteq[M]$ 
with $|S|=m$, $\p_S^A$ is far from $U_M$ in total variation distance. 
In particular, by Lemma~\ref{lem:TV-ising-lower-bound}, we have that
\begin{align*}
\dtv(U_M,\p_S^A)\ge\dtv\Big(U_M,Q_M^{S,\frac{\delta}{m}}\Big)-\dtv\Big(\p_S^A,Q_M^{S,\frac{\delta}{m}}\Big)
\ge\Omega(\delta)-O(\epsilon)=\Omega(\delta) \;.
\end{align*}
In addition, by our choice of $m$, we have that $\sqrt{2\tau}\le O(\delta)$. 
Therefore, we have that $\dtv\left(U_M,\p_S^A\right)\ge2\sqrt{2\tau}+\Omega(\delta)$. 
Applying Claim~\ref{clm:learning-to-testing} and 
Theorem~\ref{thm:sq-ising-testing} yields Theorem~\ref{thm:sq-ising}.
\end{proof}

\subsection{Proof of Proposition~\ref{prop:ising-moment-matching}}\label{ssec:ising-moment-matching}
In this section, we prove Proposition~\ref{prop:ising-moment-matching}.
We first introduce the following notations which will be used throughout this section.
For some fixed positive integer $n$ and $x\in[n]\cup\{0\}$, we consider the function $G_{n,x}(\delta)=\is(n,\delta)(x),-1/n<\delta<1/n$.
By definition, we have that $\is(n,\delta)(x)=\is(n,\delta)(n-x)$ and $G_{n,x}(\delta)=G_{n,n-x}(\delta), x\in[n]\cup\{0\}, -1/n<\delta<1/n$. In particular, $\is(n,0)$ is exactly the binomial distribution $\bin(n,1/2)$.
\begin{claim}\label{clm:h-center}
Let $n\in\Z_+$ and $X\sim\is(n,0)$. Then, $\E_{X\sim\is(n,0)}[h(n,X)]=0$.
\end{claim}
\begin{proof}
By definition, we have that
\begin{align*}
&\quad\E_{X\sim\is(n,0)}[h(n,X)]=\E_{X\sim\bin(n,1/2)}\left[2X^2-2nX+n(n-1)/2\right]\\&=2\left(\E_{X\sim\bin(n,1/2)}[X]^2+\var_{X\sim\bin(n,1/2)}[X]\right)-2n\E_{X\sim\bin(n,1/2)}[X]+n(n-1)/2\\&=2\left(n^2/4+n/4\right)-n^2+n(n-1)/2=0.
\end{align*}
\end{proof}

The first and second derivatives of $G_{n,x}(\delta)$ are given by the following claim:
\begin{claim}\label{clm:G-derivative}
Let $n\in\Z_+$ and $x\in[n]\cup\{0\}$. For any $-1/n<\delta<1/n$, we have that
\begin{align*}
G'_{n,x}(\delta)&=G_{n,x}(\delta)\left(h(n,x)-\E_{Y\sim\is(n,\delta)}[h(n,Y)]\right),\\
G''_{n,x}(\delta)&=G_{n,x}(\delta)\left((h(n,x)-\E_{Y\sim\is(n,\delta)}[h(n,Y)])^2-\var_{Y\sim\is(n,\delta)}[h(n,Y)]\right).
\end{align*}
\end{claim}
\begin{proof}
By definition, we have that $Z_n(\delta)=\sum_{x=0}^{n}\binom{n}{x}\exp(h(n,x)\delta)$ and $Z'_n(\delta)=\sum_{x=0}^{n}\binom{n}{x}h(n,x)\exp(h(n,x)\delta)$. Therefore,
\begin{align*}
G'_{n,x}(\delta)&=\binom{n}{x}\left(\frac{h(n,x)\exp(h(n,x))}{Z_n(\delta)}-\frac{\exp(h(n,x)\delta)Z'_n(\delta)}{Z_n(\delta)^2}\right)\\&=\binom{n}{x}\left(\frac{h(n,x)\exp(h(n,x))}{Z_n(\delta)}-\left(\frac{\exp(h(n,x)\delta)}{Z_n(\delta)}\right)\left(\frac{\sum_{y=0}^{n}\binom{n}{y}h(n,y)\exp(h(n,y)\delta)}{Z_{n}(\delta)}\right)\right)\\&=G_{n,x}(\delta)\left(h(n,x)-\E_{Y\sim\is(n,\delta)}[h(n,Y)]\right).
\end{align*}
For the second derivative, applying the above result for the first derivative yields
\begin{align*}
G''_{n,x}(\delta)&=G'_{n,x}(\delta)\left(h(n,x)-\E_{Y\sim\is(n,\delta)}[h(n,Y)]\right)+G_{n,x}(\delta)\frac{d}{d\delta}(\E_{Y\sim\is(n,\delta)}[h(n,Y)])\\&=G_{n,x}(\delta)\Bigg(\left(h(n,x)-\E_{Y\sim\is(n,\delta)}[h(n,Y)]\right)^2+\sum_{y=0}^{n}h(n,y)G'_{n,y}(\delta)\Bigg)\\&=G_{n,x}(\delta)\Bigg(\left(h(n,x)-\E_{Y\sim\is(n,\delta)}[h(n,Y)]\right)^2+\sum_{y=0}^{n}h(n,y)G_{n,y}(\delta)\left(h(n,y)-\E_{Y\sim\is(n,\delta)}[h(n,Y)]\right)\Bigg)\\&=G_{n,x}(\delta)\left(\left(h(n,x)-\E_{Y\sim\is(n,\delta)}[h(n,Y)]\right)^2+\E_{Y\sim\is(n,\delta)}[h(n,Y)^2]-\E_{Y\sim\is(n,\delta)}[h(n,Y)]^2\right)\\&=G_{n,x}(\delta)\left((h(n,x)-\E_{Y\sim\is(n,\delta)}[h(n,Y)])^2-\var_{Y\sim\is(n,\delta)}[h(n,Y)]\right).
\end{align*}
\end{proof}

The following claim states that for any sufficiently small parameter $\delta\ge0$, $X-n/2$ and $h(n,X)$ have sharp sub-Gaussian and sub-exponential tail, respectively.
\begin{claim}\label{clm:ising-Gaussian-exp}
Let $n\in\Z_+$. There exists universal constants $C_1,C_2>0$ such that, for any $0\le\delta\le\frac{1}{2n}$, we have that $\|X-n/2\|_{\psi_2}\le C_1\sqrt{n}$ and $\|h(n,X)\|_{\psi_1}\le C_2n$, where $X\sim\is(n,\delta)$.
\end{claim}
\begin{proof}
We consider the Ising model $P_\theta$, where $\theta_{ij}=\delta,\forall i,j\in[n],i\ne j$. Since $0\le\delta\le\frac{1}{2n}$, we have that $\sum_{j\in[n],j\ne i}|\theta_{ij}|\le1/2,\forall i\in[n]$. Let $X\sim\is(n,\delta)$. By definition, we know that $X$ denotes the number of $1$'s in the random vector of $P_\theta$. Therefore, applying Fact~\ref{fact:ising-sub-Gaussian} by taking $b$ to be the all-ones vector, we have that $\|X-n/2\|_{\psi_2}\le C_1\sqrt{n}$ for some universal constant $C_1>0$. Similarly, applying Fact~\ref{fact:ising-sub-exponential} by taking $A$ to be the all-ones matrix, we have that $\|h(n,X)\|_{\psi_1}\le C_2n$ for some universal constant $C_2>0$.
\end{proof}

We pick $C=\Theta(\sqrt{(\log(1/\delta)/m)})$, where the hidden constant is sufficiently small and consider the interval $I_{C,m}=[(1/2-C)m,(1/2+C)m-1]$. Without loss of generality, we assume that the two endpoints of $I_{C,m}$ are integers.
We define the one-dimensional distribution $A$ to be:
\begin{itemize}
\item For $x\notin I_{C,m}$, we define $A(x)=\is(m,\delta/m)(x)$.
\item For $x\in I_{C,m}$, we define $A(x)=\is(m,\delta/m)(x)+\int_{x}^{x+1}p(t)dt$, where $p$ is a polynomial of degree at most $k$ satisfying
\begin{align}\label{eq:ising-moment-matching}
\sum_{x\in\Z\cap I_{C,m}}x^i\int_{x}^{x+1}p(t)dt=\sum_{x=0}^{m}(\is(m,0)(x)-\is(m,\delta/m)(x))x^i,
\end{align}
for $0\le i\le k$.
\end{itemize}
Applying Theorem~\ref{thm:moment-matching} with the family of functions $\{G_{m,x}(\delta)\}_{x\in[m]\cup\{0\}}$, we know that there is a unique real polynomial $p$ of degree at most $k$ satisfying the above properties.
Then we need to show that with sufficiently large $m$ (depending on $\delta$), both the $L_1$ and $L_\infty$ norms of $p$ on $[(1/2-C)m,(1/2+C)m]$ are sufficiently small in order to make $A(x)$ non-negative. 
The main technical result of this section is the following lemma, which provides upper bounds on the $L_1$ and $L_\infty$ norms of $p$ on the interval $[(1/2-C)m,(1/2+C)m]$.

\begin{lemma}\label{lem:ising-norm-bound}
Let $1\le k^2\le C_0C^2m$ for some universal constant $C_0>0$ sufficiently small and $m\ge C_1(\log(1/\delta))^3$ for some universal constant $C_1>0$ sufficiently large. Then $\int_{(1/2-C)m}^{(1/2+C)m}|p(t)|dt\le O\left(\frac{\delta k^3}{C^2m}\right)$ and $|p(t^*)|\le O\left(\frac{\delta k^{7/2}}{C^3m^{2}}\right)$, where $t^*=\arg\max_{|t-m/2|\le Cm}|p(t)|$.
\end{lemma}

Before we prove Lemma~\ref{lem:ising-norm-bound}, we first use it to prove our main Proposition~\ref{prop:ising-moment-matching}.
The following lemma gives both the lower and upper bound of the ratio between the mass of $\is(m,0)$ and $\is(m,\delta/m)$.
\begin{lemma}\label{lem:ising-moment}
Let $m\in\Z_+$ and $x\in[m]\cup\{0\}$. There is a universal constant $\delta_0>0$ such that for any $0\le\delta\le\delta_0$, we have that
\begin{align*}
e^{-\delta^2/\delta_0^2}\cdot\exp(h(m,x)\delta/m)\le\frac{G_{m,x}(\delta/m)}{G_{m,x}(0)}\le\exp(h(m,x)\delta/m).
\end{align*}
\end{lemma}
\begin{proof}
By definition, we have that
\begin{align*}
\frac{G_{m,x}(\delta/m)}{G_{m,x}(0)}&=\frac{\exp(h(m,x)\delta/m)}{2^{-m}Z_{m}(\delta/m)}=\frac{\exp(h(m,x)\delta/m)}{2^{-m}\sum_{y=0}^{m}\binom{m}{y}\exp(h(m,y)\delta/m)}\\&=\frac{\exp(h(m,x)\delta/m)}{\E_{Y\sim\is(m,0)}[\exp(h(m,Y)\delta/m)]}.
\end{align*}
The upper bound is due to Claim~\ref{clm:h-center} and Jenson's inequality that $$\E_{Y\sim\is(m,0)}[\exp(h(m,Y)\delta/m)]\ge\exp\left(\E_{Y\sim\is(m,0)}[h(m,Y)\delta/m]\right)=1.$$
To prove the lower bound, by Claim~\ref{clm:ising-Gaussian-exp}, we have that $\|h(m,Y)\|_{\psi_1}\le O(m)$. Therefore, by Fact~\ref{fact:center-Gaussian-exp}, there is a universal constant $\delta_0>0$ such that for every $0\le\delta\le\delta_0$, we have that $\E_{Y\sim\is(m,0)}[\exp(h(m,Y)\delta/m)]\le\exp((m^2/\delta_0^2)(\delta^2/m^2))=\exp(\delta^2/\delta_0^2)$.
\end{proof}

We now bound from above the desired $\chi^2$-divergence:

\begin{lemma}\label{lem:ising-chi}
We have that
\begin{align*}
\chi^2(A,\mathrm{Bin}(m,1/2))\le O\left(\delta^2+\frac{\delta k^3\exp(2\delta C^2m)}{C^2m}+\left(\frac{\delta k^3}{C^2m}\right)\cdot\max_{x\in\Z\cap I_{C,m}}\frac{\int_{x}^{x+1}|p(t)|dt}{G_{m,x}(0)}\right).
\end{align*}
\end{lemma}
\begin{proof}
We have the following:
\begin{align*}
&\quad1+\chi^2(A,\mathrm{Bin}(m,1/2))=\sum_{x=0}^{m}\frac{A(x)^2}{G_{m,x}(0)}=\sum_{x=0}^{m}\frac{\left(G_{m,x}(\delta/m)+\mathbb{I}[x\in I_{C,m}]\int_{x}^{x+1}p(t)dt\right)^2}{G_{m,x}(0)}\\&=\sum_{x=0}^{m}\frac{G_{m,x}(\delta/m)^2}{G_{m,x}(0)}+2\sum_{x\in\Z\cap I_{C,m}}\frac{G_{m,x}(\delta/m)\int_{x}^{x+1}p(t)dt}{G_{m,x}(0)}+\sum_{x\in\Z\cap I_{C,m}}\frac{\left(\int_{x}^{x+1}p(t)dt\right)^2}{G_{m,x}(0)}.
\end{align*}
For the first term, by Lemma~\ref{lem:ising-moment}, we have that
\begin{align*}
\sum_{x=0}^{m}\frac{G_{m,x}(\delta/m)^2}{G_{m,x}(0)}&\le\sum_{x=0}^{m}G_{m,x}(0)\exp(2h(m,x)\delta/m)=\E_{X\sim\is(m,0)}[\exp(2h(m,X)\delta/m)]\\&\le\exp(O(\delta^2))\le1+O(\delta^2).
\end{align*}
For the second term, by Lemma~\ref{lem:ising-moment} and Lemma~\ref{lem:ising-norm-bound}, we have that
\begin{align*}
&\quad\sum_{x\in\Z\cap I_{C,m}}\frac{G_{m,x}(\delta/m)\int_{x}^{x+1}p(t)dt}{G_{m,x}(0)}\le\sum_{x\in\Z\cap I_{C,m}}\exp(h(m,x)\delta/m)\cdot\left|\int_{x}^{x+1}p(t)dt\right|\\&\le\exp\left((2C^2m-1/2)\delta\right)\sum_{x\in\Z\cap I_{C,m}}\left|\int_{x}^{x+1}p(t)dt\right|\le\exp(2\delta C^2m)\int_{(1/2-C)m}^{(1/2+C)m}|p(t)|dt\\&\le O\left(\frac{\delta k^3\exp(2\delta C^2m)}{C^2m}\right),
\end{align*}
where the second inequality follows from the fact that $h(m,x)=2x^2-2mx+\frac{m(m-1)}{2}$ attains its maximum at $x=(1/2-C)m$ over the interval $I_{C,m}$.
For the third term, by Lemma~\ref{lem:ising-norm-bound}, we have that
\begin{align*}
\sum_{x\in\Z\cap I_{C,m}}\frac{\left(\int_{x}^{x+1}p(t)dt\right)^2}{G_{m,x}(0)}&\le\int_{(1/2-C)m}^{(1/2+C)m}|p(t)|dt\cdot\max_{x\in\Z\cap I_{C,m}}\frac{\int_{x}^{x+1}|p(t)|dt}{G_{m,x}(0)}\\&\le O\left(\frac{\delta k^3}{C^2m}\right)\cdot\max_{x\in\Z\cap I_{C,m}}\frac{\int_{x}^{x+1}|p(t)|dt}{G_{m,x}(0)}.
\end{align*}
Combining the above results together completes the proof.
\end{proof}

We are now ready to prove Proposition~\ref{prop:ising-moment-matching}. We need to pick $C$ appropriately and check the bounds on $k$ needed for $A(x)$ to satisfy the necessary properties.

\begin{proof}[Proof of Proposition~\ref{prop:ising-moment-matching}]
Let $C=\Theta(\sqrt{\log(1/\delta)/m})$. If $k^3\ge C^2m$, we pick $A=\mathrm{Bin}(m,1/2)$ and obtain $\dtv(A,G_{m,\delta/m})\le O(\delta)\le O\left(\frac{\delta k^3}{\log(1/\delta)}\right)$. Thus, we assume that $k^3\le C^2m$.
In this way, to apply Lemma~\ref{lem:bin-norm-bound}, we need $k^2\le C_0C^2m$ for some universal constant $C_0$ sufficiently small, which will be satisfied as long as $\delta\le\exp(-1/C_0^3)$.

We first show that $A(x)$ is indeed a distribution over $[m]\cup\{0\}$. By definition, $A(x)$ is nonnegative outside the interval $I_{C,m}$. For $x\in\Z\cap I_{C,m}$, we apply Fact~\ref{fact:bin-bound-2}, Lemma~\ref{lem:bin-norm-bound} and Lemma~\ref{lem:ising-moment} to obtain
\begin{align*}
A(x)&=G_{m,x}(\delta/m)+\mathbb{I}[x\in I_{C,m}]\int_{x}^{x+1}p(t)dt\\&\ge e^{-O(\delta^2)}\exp(h(m,x)\delta/m)G_{m,x}(0)-|p(t^*)|=\frac{\binom{m}{x}\exp\left(h(m,x)\delta/m\right)}{2^{m}\exp(O(\delta^2))}-|p(t^*)|\\&\ge\sqrt{\frac{m}{8x(m-x)}}2^{mH(x/m)}\cdot\frac{\exp\left((2x^2/m-2x+(m-1)/2)\delta\right)}{2^m\exp(O(\delta^2))}-O\left(\frac{\delta k^{7/2}}{C^3m^{2}}\right)\\&\ge\sqrt{\frac{1}{2m}}\cdot2^{mH(q)}\cdot\frac{\exp\left((2mq^2-2mq+(m-1)/2)\delta\right)}{2^m\exp(O(\delta^2))}-O\left(\frac{\delta k^{7/2}}{C^3m^{2}}\right)\\&\ge\Omega\left(\sqrt{\frac{1}{m}}\right)\cdot2^{m\left(H(q)+2\delta(q^2-q)-1\right)}\cdot\exp\left((m-1)\delta/2\right)-O\left(\frac{\delta k^{7/2}}{C^3m^{2}}\right),
\end{align*}
where we let $q=x/m$ and apply the fact $x(m-x)\le m^2/4$ in the third inequality. Let $f(q)=H(q)+\lambda(q^2-q)-1$, where $\lambda=2\delta$. We have that $f'(q)=\log_2\left(\frac{1-q}{q}\right)+(2q-1)\lambda$ and $f''(q)=2\lambda-\frac{1}{q(1-q)}\le2\lambda-4<0$ as long as $\lambda<2$, which implies that $f(q)$ is strictly concave over $[0,1]$ and attains its maximum at $q=1/2$. Therefore, we have that
\begin{align*}
A(x)&\ge\Omega\left(\sqrt{\frac{1}{m}}\right)\cdot2^{mf(x/m)}\cdot\exp\left((m-1)\delta/2\right)-O\left(\frac{\delta k^{7/2}}{C^3m^{2}}\right)\\&\ge\Omega\left(\sqrt{\frac{1}{m}}\right)\cdot2^{mf(1/2+C)}\cdot\exp\left((m-1)\delta/2\right)-O\left(\frac{\delta k^{7/2}}{C^3m^{2}}\right)\\&\ge\Omega\left(\sqrt{\frac{1}{m}}\right)\cdot2^{m(H(1/2+C)-1)}\cdot2^{(2mC^2-1/2)\delta}-O\left(\frac{\delta k^{7/2}}{C^3m^{2}}\right)\\&\ge\sqrt{\frac{1}{m}}\cdot\exp\left(-O(C^2m)\right)-O\left(\frac{\delta k^{7/2}}{C^3m^2}\right),
\end{align*}
where the last inequality follows from the Taylor expansion of $H(1/2+C)-H(1/2)$ up to second order terms.
Note that $k^3\le C^2m$, by our choice of $C$, where $C=\Theta(\sqrt{\log(1/\delta)/m})$ for some sufficiently small hidden constant in $\Theta$, we have that $\frac{\delta k^{7/2}}{C^3m^{3/2}}\le O(\delta(\log(1/\delta))^{-1/3})$ and $\exp(-O(C^2m))\ge\delta$.
Therefore, we have that
$$A(x)\ge\sqrt{\frac{1}{m}}\cdot\exp\left(-O(C^2m)\right)-O\left(\frac{\delta k^{7/2}}{C^3m^2}\right)\ge0,\forall x\in\Z\cap I_{C,m}.$$

In addition, by Equation~\eqref{eq:ising-moment-matching}, we know that
\begin{align*}
\sum_{x=0}^{m}A(x)&=\sum_{x=0}^{m}\left(G_{m,x}(\delta/\sqrt{m})+\mathbb{I}[x\in I_{C,m}]\int_{x}^{x+1}p(t)dt\right)\\&=\sum_{x=0}^{m}G_{m,x}(\delta/\sqrt{m})+\int_{(1/2-C)m}^{(1/2+C)m}p(t)dt=1,
\end{align*}
which implies that the distribution $A$ is well-defined. Furthermore, by Equation~\eqref{eq:ising-moment-matching}, we can show that $A$ matches the first $k$ moments of $\bin(m,1/2)$ as follows:
\begin{align*}
\E_{X\sim A}[X^i]&=\sum_{x=0}^{m}{A(x)x^i}=\sum_{x=0}^{m}\left(G_{m,x}(\delta/\sqrt{m})+\mathbb{I}[x\in I_{C,m}]\int_{x}^{x+1}p(t)dt\right)x^i\\&=\sum_{x=0}^{m}G_{m,x}(\delta/\sqrt{m})x^i+\sum_{x\in\Z\cap I_{C,m}}x^i\int_{x}^{x+1}p(t)dt\\&=\sum_{x=0}^{m}G_{m,x}(0)x^i=\E_{X\sim\bin(m,1/2)}[X^i].
\end{align*}
From the previous calculations, we have that $A(x)\ge e^{-O(\delta^2)}\exp(h(m,x)\delta/m)G_{m,x}(0)-|p(t^*)|\ge0,\forall x\in\Z\cap I_{C,m}$, which implies that for every $x\in\Z\cap I_{C,m}$,
\begin{align*}
|p(t^*)|\le e^{-O(\delta^2)}\exp(h(m,x)\delta/m)G_{m,x}(0)\le\exp(2\delta C^2m)G_{m,x}(0).
\end{align*}
Therefore, by Lemma~\ref{lem:ising-chi}, we have that
\begin{align*}
\chi^2(A,\mathrm{Bin}(m,1/2))&\le O\left(\delta^2+\frac{\delta k^3\exp(2\delta C^2m)}{C^2m}+\left(\frac{\delta k^3}{C^2m}\right)\cdot\max_{x\in\Z\cap I_{C,m}}\frac{\int_{x}^{x+1}|p(t)|dt}{G_{m,x}(0)}\right)\\&\le O\left(\delta^2+\delta\left(\exp(2\delta C^2m)+\frac{|p(t^*)|}{G_{m,x}(0)}\right)\right)\le O\left(\delta^2+2\delta\exp(2\delta C^2m)\right)\\&\le O\left(\delta^2+\delta(1+O(\delta\log(1/\delta)))\right)=O(\delta),
\end{align*}
where the last inequality follows from the fact $e^x\le1+2x,\forall x\in[0,\ln2]$.

Finally, to bound the total variation distance $\dtv(A,\is(m,\delta/m))$, we apply Lemma~\ref{lem:ising-norm-bound} to obtain
\begin{align*}
\dtv(A,\is(m,\delta/m)))&=\sum_{x\in\Z\cap I_{C,m}}\left|\int_{x}^{x+1}p(t)dt\right|\le\int_{(1/2-C)m}^{(1/2+C)m}|p(t)|dt\\&\le O\left(\frac{\delta k^3}{C^2m}\right)=O\left(\frac{\delta k^3}{\log(1/\delta)}\right).
\end{align*}
\end{proof}

\paragraph{Proof of Lemma~\ref{lem:ising-norm-bound}}
By Theorem~\ref{thm:moment-matching}, we have that
\begin{align*}
|a_i|\le\left(\frac{2i+1}{2Cm}\right)\left(\beta_i+O\left(\frac{i^2}{Cm}\right)\int_{(1/2-C)m}^{(1/2+C)m}|p(t)|dt\right),
\end{align*}
for all $1\le i\le k$, where $\beta_i=\left|\sum_{x=0}^{m}(G_{m,x}(0)-G_{m,x}(\delta/m))P_i\left(\frac{x-m/2}{Cm}\right)\right|$. To get an upper bound for the $L_1$ and $L_\infty$ norms of the polynomial $p$ over $[(1/2-C)m,(1/2+C)m]$, we only need to upper bound the quantity $\beta_i$.
\begin{lemma}\label{lem:ising-beta}
If $k^2\le C_0C^2m$ for some universal constant $C_0>0$ sufficiently small, then $\beta_i\le O\left(\frac{\delta i^{3/2}}{C^2m}\right),\forall1\le i\le k$.
\end{lemma}

We assume $k^2\le C_0C^2m$ for some universal constant $C_0>0$ sufficiently small. Note that by our definition $G_{m,x}(\delta)=G_{m,m-x}(\delta),\forall-1/2<\delta<1/2$, by Fact~\ref{fact:legendre-poly} (iv), we have that $\beta_i=\left|\sum_{x=0}^{m}(G_{m,x}(0)-G_{m,x}(\delta/m))P_i\left(\frac{x-m/2}{Cm}\right)\right|=0$ for any odd $i$.
Hence, we only need to bound $\beta_i$ for every even $i$. We apply Taylor's theorem to expand $G_{m,x}(\delta/m)-G_{m,x}(0)$ up to second order terms:
\begin{align}\label{eq:G-taylor}
\beta_i&=\left|\sum_{x=0}^{m}\left(G_{m,x}(0)-G_{m,x}(\delta/\sqrt{m})\right)P_i\left(\frac{x-m/2}{Cm}\right)\right|\notag\\&=\left|\sum_{x=0}^{m}\left(\frac{\delta G'_{m,x}(0)}{\sqrt{m}}+\frac{G''_{m,x}(\delta_x)\delta^2}{2m}\right)P_i\left(\frac{x-m/2}{Cm}\right)\right|\notag\\&\le\underbrace{\left|\sum_{x=0}^{m}\left(\frac{\delta G'_{m,x}(0)}{\sqrt{m}}\right)P_i\left(\frac{x-m/2}{Cm}\right)\right|}_{\beta'_i}+\underbrace{\left|\sum_{x=0}^{m}\left(\frac{G''_{m,x}(\delta_x)\delta^2}{2m}\right)P_i\left(\frac{x-m/2}{Cm}\right)\right|}_{\beta''_i},
\end{align}
where for any $x\in[m]\cup\{0\}$, $\delta_x=\wt{\delta_x}/\sqrt{m}$ for some $\wt{\delta_x}\in[0,\delta]$.

Hence, in order to bound $\beta_i$, it suffices to bound the terms $\beta'_i:=\frac{\delta}{\sqrt{m}}\left|\sum_{x=0}^{m}G'_{m,x}(0)P_i\left(\frac{x-m/2}{Cm}\right)\right|$ and $\beta_i'':=\frac{\delta^2}{2m}\left|\sum_{x=0}^{m}G''_{m,x}(\delta_x)P_i\left(\frac{x-m/2}{Cm}\right)\right|$. This is done in the following lemmas.

\begin{lemma}\label{lem:ising-first-order}
For every even $i$, we have that $\beta'_i\le O\left(\frac{\delta i^{3/2}}{C^2m}\right)$.
\end{lemma}

\begin{proof}
By Claim~\ref{clm:h-center} and Claim~\ref{clm:G-derivative}, we have that $G'_{m,x}(0)=G_{m,x}(0)h(m,x)$. Applying Claim~\ref{clm:h-center} and Fact~\ref{fact:legendre-poly} (v) by using the change of variables yields
\begin{align*}
\beta'_i&=\left|\sum_{x=0}^{m}\left(\frac{\delta G'_{m,x}(0)}{m}\right)P_i\left(\frac{x-m/2}{Cm}\right)\right|=\frac{\delta}{m}\left|\sum_{x=0}^{m}G_{m,x}(0)h(m,x)P_i\left(\frac{x-m/2}{Cm}\right)\right|\\&=\frac{\delta}{m}\left|\sum_{x=0}^{m}G_{m,x}(0)h(m,x)2^{-i}\sum_{j=0}^{i/2}(-1)^j\binom{i}{j}\binom{2i-2j}{i}\left(\frac{x-m/2}{Cm}\right)^{i-2j}\right|\\&=\frac{\delta}{m}\left|\sum_{x=0}^{m}G_{m,x}(0)h(m,x)2^{-i}\sum_{j=0}^{i/2-1}(-1)^j\binom{i}{j}\binom{2i-2j}{i}\left(\frac{x-m/2}{Cm}\right)^{i-2j}\right|\\&\le\frac{\delta}{m}\sum_{x=0}^{m}G_{m,x}(0)\left|h(m,x)\right|2^{-i}\sum_{j=1}^{i/2}\binom{i}{i/2+j}\binom{i+2j}{2j}\left|\frac{x-m/2}{Cm}\right|^{2j}\\&\le\frac{\delta}{m}\sum_{j=1}^{i/2}\sum_{x=0}^{m}G_{m,x}(0)\left|h(m,x)(x-m/2)^{2j}\right|O\left(\frac{1}{\sqrt{i}}\right)\frac{(i+2j)^{2j}}{(2j)!(Cm)^{2j}}.
\end{align*}
Since $\|X-m/2\|_{\psi_2}\le O(\sqrt{m})$ for $X\sim\bin(m,1/2)$, applying Fact~\ref{fact:Gaussian-decay} and Fact~\ref{fact:bin-bound-1} yields
\begin{align*}
\beta_i&\le\frac{\delta}{m}\sum_{j=1}^{i/2}\sum_{x=0}^{m}G_{m,x}(0)\left|h(m,x)(x-m/2)^{2j}\right|O\left(\frac{1}{\sqrt{i}}\right)\left(\frac{(i+2j)^{2j}}{(2j)!(Cm)^{2j}}\right)\\&=\frac{\delta}{m}\sum_{j=1}^{i/2}\sum_{x=0}^{m}G_{m,x}(0)\left|(2x^2-2mx+m(m-1)/2)(x-m/2)^{2j}\right|O\left(\frac{1}{\sqrt{i}}\right)\left(\frac{(i+2j)^{2j}}{(2j)!(Cm)^{2j}}\right)\\&\le\frac{\delta}{m}\sum_{j=1}^{i/2}\sum_{x=0}^{m}G_{m,x}(0)\left(2(x-m/2)^{2j+2}+m(x-m/2)^{2j}/2\right)O\left(\frac{1}{\sqrt{i}}\right)\left(\frac{(i+2j)^{2j}}{(2j)!(Cm)^{2j}}\right)\\&=\frac{\delta}{m}\sum_{j=1}^{i/2}O\left(\frac{1}{\sqrt{i}}\right)\left(\frac{(i+2j)^{2j}\E_{X\sim\bin(m,1/2)}[2(X-m/2)^{2j+2}+m(X-m/2)^{2j}/2]}{(2j)!(Cm)^{2j}}\right)\\&\le O\left(\frac{\delta}{m\sqrt{i}}\right)\sum_{j=1}^{i/2}\frac{(i+2j)^{2j}}{(2j)!(Cm)^{2j}}\left((2j+2)(j!)(O(m))^{j+1}+(j!)(O(m))^{j+1}\right)\\&\le O\left(\frac{\delta}{\sqrt{i}}\right)\sum_{j=1}^{\infty}\left(O\left(\frac{i}{C\sqrt{m}}\right)\right)^{2j}\le O\left(\frac{\delta i^{3/2}}{C^2m}\right).
\end{align*}
\end{proof}

\begin{lemma}\label{lem:ising-second-order}
For every even $i$, we have that $\beta''_i\le O(\delta^2)$.
\end{lemma}
\begin{proof}
By Claim~\ref{clm:G-derivative}, we have that
\begin{align*}
\beta''_i&=\left|\sum_{x=0}^{m}\left(\frac{G''_{m,x}(\delta_x)\delta^2}{2m^2}\right)P_i\left(\frac{x-m/2}{Cm}\right)\right|\\&=\frac{\delta^2}{2m^2}\left|\sum_{x=0}^{m}G_{m,x}(\delta_x)\Bigg(\Big(h(m,x)-\mathop{\E}_{Y\sim\is(m,\delta_x)}[h(m,Y)]\Big)^2-\mathop{\var}_{Y\sim\is(m,\delta_x)}[h(m,Y)]\Bigg)P_i\left(\frac{x-m/2}{Cm}\right)\right|.
\end{align*}

We separate the above sum into $x\in\Z\cap I_{C,m}$ and $x\in[m]\cup\{0\}\setminus I_{C,m}$. We are able to use Fact~\ref{fact:legendre-poly} (iii) to bound the sum for $x\in\Z\cap I_{C,m}$, as follows:
\begin{align*}
&\quad\sum_{x\in\Z\cap I_{C,m}}G_{m,x}(\delta_x)\Big|\Big(h(m,x)-\mathop{\E}_{Y\sim\is(m,\delta_x)}[h(m,Y)]\Big)^2-\mathop{\var}_{Y\sim\is(m,\delta_x)}[h(m,Y)]\Big|\left|P_i\left(\frac{x-m/2}{Cm}\right)\right|\\&\le\sum_{x=0}^{m}G_{m,x}(\delta_x)\Big|\Big(h(m,x)-\mathop{\E}_{Y\sim\is(m,\delta_x)}[h(m,Y)]\Big)^2-\mathop{\var}_{Y\sim\is(m,\delta_x)}[h(m,Y)]\Big|\\&\le2\var_{Y\sim\is(m,\delta_x)}[h(m,Y)]\le O(m^2),
\end{align*}
where the last inequality follows from Fact~\ref{fact:exp-decay}, Fact~\ref{fact:Gaussian-exp-2} and Claim~\ref{clm:ising-Gaussian-exp}.

Now we bound the sum over $x\in\overline{I_{C,m}}$, where $\overline{I_{C,m}}=[m]\cup\{0\}\setminus I_{C,m}$. Note that for $X\sim\is(m,\delta_x)$, we have that $\left\|\frac{4(X-m/2)}{Cm}\right\|_{\psi_2}\le O\left(\frac{1}{C\sqrt{m}}\right)$ and $\left\|h(m,X)\right\|_{\psi_1}\le O(m)$. Therefore, applying Fact~\ref{fact:legendre-poly} (vi) yields
\begin{align*}
&\quad\sum_{x\in\overline{I_{C,m}}}G_{m,x}(\delta_x)\Big|\Big(h(m,x)-\mathop{\E}_{Y\sim G_{m,\delta_x}}[h(m,Y)]\Big)^2-\mathop{\var}_{Y\sim\is(m,\delta_x)}[h(m,Y)]\Big|\left|P_i\left(\frac{x-m/2}{Cm}\right)\right|\\&\le\sum_{x=0}^{m}G_{m,x}(\delta_x)\Big|\Big(h(m,x)-\mathop{\E}_{Y\sim\is(m,\delta_x)}[h(m,Y)]\Big)^2-\mathop{\var}_{Y\sim\is(m,\delta_x)}[h(m,Y)]\Big|\left(\frac{4|x-m/2|}{Cm}\right)^i\\&\le\mathop{\E}_{X\sim\is(m,\delta_x)}\left[\Big(h(m,X)-\mathop{\E}_{Y\sim\is(m,\delta_x)}[h(m,Y)]\Big)^2\left(\frac{4|X-m/2|}{Cm}\right)^i\right]\\&\quad+O(m^2)\cdot\mathop{\E}_{X\sim\is(m,\delta_x)}\left[\left(\frac{4|X-m/2|}{Cm}\right)^i\right]\\&\le\sqrt{\mathop{\E}_{X\sim\is(m,\delta_x)}\left[\Big(h(m,X)-\mathop{\E}_{Y\sim\is(m,\delta_x)}[h(m,Y)]\Big)^4\right]\cdot\mathop{\E}_{X\sim\is(m,\delta_x)}\left[\left(\frac{4|X-m/2|}{Cm}\right)^{2i}\right]}\\&\quad+O(m^2)\cdot\mathop{\E}_{X\sim\is(m,\delta_x)}\left[\left(\frac{4|X-m/2|}{Cm}\right)^i\right]\\&\le\sqrt{O(m^4)\left(O\left(\frac{i}{C^2m}\right)\right)^i}+O(m^2)\left(O\left(\frac{1}{C}\sqrt{\frac{i}{m}}\right)\right)^i\le O(m^2),
\end{align*}
where the third inequality follows from Cauchy-Schwarz and the fourth inequality follows from Fact~\ref{fact:Gaussian-decay}, Fact~\ref{fact:exp-decay}, Fact~\ref{fact:Gaussian-exp-2} and Claim~\ref{clm:ising-Gaussian-exp}.
Combine the above results together, we have that 
\begin{align*}
\beta_i''=\frac{\delta^2}{2m^2}\left|\sum_{x=0}^{m}G''_{m,x}(\delta_x)P_i\left(\frac{x-m/2}{Cm}\right)\right|\le\left(\frac{\delta^2}{2m^2}\right)\cdot O(m^2)=O(\delta^2).
\end{align*}
\end{proof}

Now we are ready to prove Lemma~\ref{lem:ising-norm-bound}.
\begin{proof}[Proof of Lemma~\ref{lem:ising-norm-bound}]
By Theorem~\ref{thm:moment-matching}, we have that
\begin{align*}
|p(t^*)|&\le\sum_{i=1}^{k}|a_i|\le\sum_{i=1}^{k}\left(\frac{2i+1}{2Cm}\right)\beta_i+\sum_{i=1}^{k}\left(\frac{2i+1}{2Cm}\right)O\left(\frac{i^2}{Cm}\right)\int_{(1/2-C)m}^{(1/2+C)m}|p(t)|dt\\&\le\sum_{i=1}^{k}\left(\frac{2i+1}{2Cm}\right)\beta_i+|p(t^*)|\sum_{i=1}^{k}\left(\frac{2i+1}{2}\right)O\left(\frac{i^2}{Cm}\right)\\&\le\sum_{i=1}^{k}\left(\frac{2i+1}{2Cm}\right)\beta_i+O\left(\frac{k^4}{Cm}\right)|p(t^*)|,
\end{align*}
where the first inequality follows from Fact~\ref{fact:legendre-poly} (iii). Similarly, by Theorem~\ref{thm:moment-matching}, we have that
\begin{align*}
\int_{(1/2-C)m}^{(1/2+C)m}|p(t)|dt&\le\sum_{i=1}^{k}|a_i|\int_{(1/2-C)m}^{(1/2+C)m}\left|P_i\left(\frac{t-m/2}{Cm}\right)\right|dt\\&\le Cm\sum_{i=1}^{k}\left(\frac{2i+1}{2Cm}\right)\left(\beta_i+O\left(\frac{i^2}{Cm}\right)\int_{(1/2-C)m}^{(1/2+C)m}|p(t)|dt\right)\int_{-1}^{1}|P_i(y)|dy\\&\le\sum_{i=1}^{k}O(\sqrt{i})\beta_i+\sum_{i=1}^{k}O(\sqrt{i})O\left(\frac{i^2}{Cm}\right)\int_{(1/2-C)m}^{(1/2+C)m}|p(t)|dt\\&\le\sum_{i=1}^{k}O(\sqrt{i})\beta_i+O\left(\frac{k^{7/2}}{Cm}\right)\int_{(1/2-C)m}^{(1/2+C)m}|p(t)|dt,
\end{align*}
where the third inequality follows from Fact~\ref{fact:legendre-poly} (vii). By our assumption on $k,C,m,\delta$, we know that $\frac{k^{7/2}}{Cm}\le\frac{k^4}{Cm}\le1/2$. Therefore, by Lemma~\ref{lem:bin-beta}, we have that
\begin{align*}
&|p(t^*)|\le2\sum_{i=1}^{k}\left(\frac{2i+1}{2Cm}\right)\beta_i\le\sum_{i=1}^{k}\left(\frac{2i+1}{2Cm}\right)O\left(\frac{\delta i^{3/2}}{C^2m}\right)\le O\left(\frac{\delta k^{7/2}}{C^3m^2}\right),\\&
\int_{(1/2-C)m}^{(1/2+C)m}|p(x)|dx\le2\sum_{i=1}^{k}O(\sqrt{i})\beta_i\le2\sum_{i=1}^{k}O(\sqrt{i})O\left(\frac{\delta i^{3/2}}{C^2m}\right)\le O\left(\frac{\delta k^3}{C^2m}\right).
\end{align*}
This completes the proof.
\end{proof}

\subsection{Proof of Lemma~\ref{lem:TV-ising-lower-bound}}\label{ssec:TV-ising-lower-bound}
Let $n=|S|$. Recalling that $G_{n,x}(\delta)=\binom{n}{x}\exp\left(h(n,x)\delta\right)/Z_n(\delta)$, where $h(n,x)=2x^2-2nx+\frac{n(n-1)}{2}$ and $Z_n(\delta)=\sum_{x=0}^{n}\binom{n}{x}\exp(h(n,x)\delta)$. By Claim~\ref{clm:ising-Gaussian-exp}, for any $0\le\delta\le\frac{1}{2n}$, we have that $\|h(n,X)\|_{\psi_1}\le O(n)$ and $\|X-n/2\|_{\psi_2}\le O(\sqrt{n})$ for $X\sim\is(n,\delta)$.

Define $f(\bx)=\sum_{i\in S}x_i,\forall \bx\in\{0,1\}^M$. For $\bX\sim U_M$ and $\bY\sim Q_M^{S,\frac{\delta}{m}}$, by the data processing inequality, we have that
\begin{align*}
\dtv\Big(U_M,Q_M^{S,\frac{\delta}{m}}\Big)\ge\dtv(f(\bX),f(\bY))=\dtv(\is(n,0),\is(n,\delta/m)).
\end{align*}
By the mean value theorem, we have that
\begin{align*}
\dtv(\is(n,0),\is(n,\delta/m))=\frac{1}{2}\sum_{x=0}^{n}\left|G_{n,x}(\delta/m)-G_{n,x}(0)\right|=\frac{1}{2}\sum_{x=0}^{n}\left|G'_{n,x}(0)(\delta/m)+\frac{G''_{n,x}(\delta_x)\delta^2}{2m^2}\right|,
\end{align*}
where for any $x\in[m]\cup\{0\}$, $\delta_x=\wt{\delta_x}/m$ for some $\wt{\delta_x}\in(0,\delta)$.

By elementary calculation, we have that
\begin{align*}
\E_{X\sim\is(n,0)}[h(n,X)^2]&=4\E_{X\sim\bin(n,1/2)}\left[\left((X-n/2)^2-\E_{X\sim\bin(n,1/2)}[(X-n/2)^2]\right)^2\right]\\&=4\left(\E_{X\sim\bin(n,1/2)}\left[(X-n/2)^4\right]-\E_{X\sim\bin(n,1/2)}\left[(X-n/2)^2\right]^2\right)\\&=\left(\frac{3n^2}{4}-\frac{n}{2}\right)-\frac{n^2}{4}=\frac{n^2-n}{2}.
\end{align*}
By Fact~\ref{fact:exp-decay}, we have that $\E_{X\sim\is(n,0)}[h(n,X)^4]\le O(n^4)$. Therefore, by Fact~\ref{fact:E-lower}, we have that
\begin{align*}
\sum_{x=0}^{n}|G'_{n,x}(0)|=\sum_{x=0}^{n}G_{n,x}(0)|h(n,x)|=\E_{X\sim\bin(n,1/2)}[|h(n,X)|]\ge\frac{\E_{X\sim\bin(n,1/2)}[h(n,X)^2]^{3/2}}{\E_{X\sim\bin(n,1/2)}[h(n,X)^4]^{1/2}}\ge\Omega(n).
\end{align*}
In addition, by Fact~\ref{fact:exp-decay}, we have that
\begin{align*}
\sum_{x=0}^{n}|G''_{n,x}(\delta_x)|=&\sum_{x=0}^{n}G_{n,x}(\delta_x)\left|(h(n,x)-\E_{Y\sim\is(n,\delta_x)}[h(n,Y)])^2-\var_{Y\sim\is(n,\delta_x)}[h(n,Y)]\right|\\&\le2\var_{X\sim\is(n,\delta_x)}[h(n,X)]\le O\left(n^2\right).
\end{align*}
Therefore, we have that
\begin{align*}
\dtv\Big(U_M,Q_M^{S,\frac{\delta}{m}}\Big)\ge\frac{1}{2}\sum_{x=0}^{n}\left|G'_{n,x}(0)(\delta/m)+\frac{G''_{n,x}(\delta_x)\delta^2}{2m^2}\right|\ge\Omega\left(\frac{\delta n}{m}\right)-O\left(\frac{n^2\delta^2}{m^2}\right)\ge\Omega(\delta)-O(\delta^2)=\Omega(\delta).
\end{align*}

\end{document}